\documentclass[journal, final]{IEEEtran}

\usepackage{cite}
\usepackage{amsmath,amssymb,amsfonts, amsthm}
\usepackage{algorithmic}
\usepackage{graphicx}
\usepackage{textcomp}
\usepackage[table,xcdraw]{xcolor}
\usepackage{pifont}
\usepackage{booktabs}
\usepackage{algorithm}
\usepackage{array}
\usepackage[caption=false]{subfig}

\usepackage{stfloats}
\usepackage{url}
\usepackage{verbatim}
\usepackage{cite}
\usepackage{soul}

\usepackage{enumerate, enumitem}

\usepackage{cleveref}
\usepackage{bbm}
\usepackage{bm}
\usepackage{siunitx}
\usepackage{comment}

\usepackage{pgfplots}

\pgfplotsset{plot coordinates/math parser=false}
\usetikzlibrary{plotmarks, decorations, backgrounds, spy}
\usepgfplotslibrary{colorbrewer}
\pgfplotsset{plot coordinates/math parser=false}
	\tikzset{new spy style/.style={spy scope={magnification=4,size=2cm, connect spies,every spy on node/.style={rectangle,draw,},
				every spy in node/.style={draw,rectangle,}}}}

\usepackage[acronym,style=alttree,nomain,shortcuts,nonumberlist,nogroupskip]{glossaries}
\newacronym{acas}{ACAS}{assisted \ac{cas}}
\newacronym{awgn}{AWGN}{additive white Gaussian noise}
\newacronym{boc}{BOC}{binary offset carrier}
\newacronym{bgd}{BGD}{broadcast group delay}

\newacronym{cas}{CAS}{commercial authentication service}
\newacronym{cdf}{CDF}{cumulative distribution function}
\newacronym{cdma}{CDMA}{code division multiple access}
\newacronym{chim}{CHIMERA}{chips-message robust authentication}
\newacronym{cs}{CS}{commercial service} 

\newacronym{det}{DET}{detection error trade-off}
\newacronym{dop}{DOP}{dilution of precision} 

\newacronym{ecs}{ECS}{encrypted code sequence}
\newacronym{ecef}{ECEF}{Earth-Centered Earth-Fixed}
\newacronym{enu}{ENU}{east, north, up}

\newacronym{fa}{FA}{false alarm}

\newacronym{gdop}{GDOP}{geometric dilution of precision}
\newacronym{glrt}{GLRT}{generalized likelihood ratio test}
\newacronym{ggto}{GGTO}{GPS to Galileo time offset}
\newacronym{gnss}{GNSS}{global navigation satellite system}
\newacronym{gps}{GPS}{global positioning system}
\newacronym{gpst}{GPST}{GPS system time}
\newacronym{gsc}{GSC}{GNSS service center}
\newacronym{gst}{GST}{Galileo system time}

\newacronym{has}{HAS}{high accuracy service}

\newacronym{imu}{IMU}{inertial measurement unit}
\newacronym{icto}{ICTO}{inter-constellation time offset}
\newacronym{iot}{IoT}{Internet of things}
\newacronym{isb}{ISB}{inter-system bias}

\newacronym{lan}{LAN}{local area network}
\newacronym{los}{LOS}{line-of-sight}
\newacronym{ls}{LS}{least squares}

\newacronym{mac}{MAC}{message authentication code}
\newacronym{md}{MD}{missed detection}
\newacronym{mse}{MSE}{mean square error}
\newacronym{nma}{NMA}{navigation message authentication}
\newacronym{nmea}{NMEA}{national marine electronics association}

\newacronym{ntp}{NTP}{network time protocol}

\newacronym{osnma}{OS-NMA}{open service \ac{nma}}

\newacronym{pdop}{PDOP}{position dilution of precision}
\newacronym{pdf}{PDF}{probability density function}
\newacronym{pnt}{PNT}{position navigation and timing}
\newacronym{ppm}{ppm}{parts per million}
\newacronym{prn}{PRN}{pseudo-random noise}
\newacronym{ptp}{PTP}{precision time protocol}

\newacronym{pvt}{PVT}{position, velocity, and time}

\newacronym{recs}{RECS}{re-encrypted code sequence}
\newacronym{rinex}{RINEX}{receiver independent exchange format}

\newacronym{sca}{SCA}{spreading code authentication}
\newacronym{sce}{SCE}{spreading code encryption}
\newacronym{scer}{SCER}{secure code estimation and replay}
\newacronym{sbf}{SBF}{Septentrio binary format}

\newacronym{sdr}{SDR}{software-defined radio}
\newacronym{sqm}{SQM}{signal quality monitoring}
\newacronym{sssc}{SSSCs}{spread spectrum security codes}
\newacronym{sv}{SV}{satellites vehicles}

\newacronym{tesla}{TESLA}{timed-efficient stream loss-tolerant authentication} 

\newacronym{uere}{UERE}{user equivalent ranging error}
\newacronym{utc}{UTC}{coordinated universal time}
\newacronym{vctcxo}{VCTCXO}{voltage-controlled and temperature-controlled crystal oscillator}

\newlength\fwidth
\newlength\fheight
\newtheorem{theorem}{Theorem}
\DeclareMathOperator*{\argmin}{arg\,min} 

\newcommand{\cmark}{\ding{51}}%
\newcommand{\xmark}{\ding{55}}%

\begin{document}

\title{On the Limits of Cross-Authentication Checks \\ for GNSS Signals}
\author{ \IEEEauthorblockN{ Francesco~Ardizzon {\em IEEE Member}, Laura~Crosara {\em IEEE Student Member}, Stefano~Tomasin {\em IEEE Senior Member}, and Nicola~Laurenti}\thanks{Department of Information Engineering (DEI), University of Padova, Italy. e-mail: \small \texttt{\{ardizzonfr, crosaralau, tomasin, nil\}@dei.unipd.it} 
}
}
\normalsize
\maketitle

\begin{abstract}
\Acp{gnss} are implementing security mechanisms: examples are Galileo open service navigation message authentication (OS-NMA) and GPS \ac{chim}. 
Each of these mechanisms operates in a single band. However, nowadays, even commercial \ac{gnss} receivers typically compute the \ac{pvt} solution using multiple constellations and signals from multiple bands at once, significantly improving both accuracy and availability. 
Hence, cross-authentication checks have been proposed, based on the \ac{pvt} obtained from the mixture of authenticated and non-authenticated signals.

In this paper, first, we formalize the models for the cross-authentication checks. Next, we describe, for each check, a spoofing attack to generate a fake signal leading the victim to a target \ac{pvt} without notice. We analytically relate the degrees of the freedom of the attacker in manipulating the victim's solution to both the employed security checks and the number of open signals that can be tampered with by the attacker. We test the performance of the considered attack strategies on an experimental dataset.
Lastly, we show the limits of the \ac{pvt}-based \ac{gnss} cross-authentication checks, where both authenticated and non-authenticated signals are used.
\end{abstract}


\sloppy 

\glsresetall
\section{Introduction}\label{sec:introduction}
\Acp{gnss} provide real-time positioning and timing for various civil and military applications. However, \ac{gnss} signals are particularly susceptible to both accidental and intentional interference due to their low received power. Furthermore, civilian \ac{gnss} signal and modulation formats are open to the public \cite{GALICD, GPSNAVSTAR}: therefore, an attacker can compromise location or timing-aware applications by generating fake \ac{gnss} signals, leading the receiver into the calculation of a fake location or time \cite{sanjaICL, spoofingSurvey}. These attacks are referred to as {\em spoofing attacks}. Additionally, several works have reported successful attacks even using off-the-shelf hardware \cite{ceccato18,lenhart21}. 

Thus, service providers proposed several authentication mechanisms. On the system side, authentication is strengthened by incorporating into the broadcast \ac{gnss} signals specific features, which cannot be predicted and generated by the attacker. An authentication-enabled receiver can then interpret these characteristics to distinguish authentic signals from forgeries. The features can be inserted  at two complementary levels: the {\em data level}, i.e, on navigation messages, and the {\em ranging level}, i.e., on pseudoranges between the satellite and receiver. 

\Ac{nma} techniques aim at ensuring the authenticity of the content of the navigation messages. \Ac{osnma} \cite{OSNMAart} is a data authentication function for public Galileo E1B signals, in which the message transmitted by the satellites is interleaved with authentication data, generated through a broadcast authentication protocol called \ac{tesla} \cite{tesla,independentTimeSync,OSNMAart2}. 
However, \ac{nma}-based protocols are vulnerable to \ac{scer} attacks \cite{HumDetStrat,improvingSCER} that estimate the message signature to forge a fake signal in real-time.

Authentication at the ranging level operates on the signal spreading code, de facto authenticating each signal travel time, i.e., the time it takes for the signal to reach the receiver. \Ac{sce} techniques are the most reliable options to limit access to \ac{gnss} signals, as they make the whole spreading code unpredictable. A \ac{sce} solution for Galileo is represented by the \ac{cas}, which will complement \ac{osnma} by offering spreading ranging level authentication in the E6 band. 
The \ac{acas}, recently presented in~\cite{acasICL} and \cite{cas22}, is instead a \ac{sca} method based on the navigation data received and authenticated by \ac{osnma}.  
For GPS, the \ac{chim} scheme, \cite{Scott2003,chimeraION}, aims at jointly authenticating the navigation data and the spreading code of GPS signals for civil usage, implementing both \ac{nma} and \ac{sca}. This solution replaces a small part of the spreading code with a secret cryptographic sequence, which can be subsequently reproduced by the receivers when they become aware of the key. Navigation message data are instead protected by digitally signing most or all the data. 
Indeed, ranging level authentication reduces the effectiveness of \ac{scer}-type attacks, as discussed in~\cite{caparra18}. Thus, in the following, we will refer to signals protected by \ac{sca} (or \ac{sce}) as {\em authenticated signals}, while all the others are {\em open signals}.

Both \ac{acas} and \ac{chim} present several issues. First, data for range verification is only provided to the receivers a posteriori, therefore they do not allow the verification of the current \ac{pvt}. To address this, in~\cite{paperSensors}, the authors proposed a strategy where only \ac{acas} is used to derive a continuous, robust, and secure timing service: the application however is limited to static scenarios. A verification method operating on a continuous GPS signal and using a stochastic reachability-based Kalman filter is proposed in~\cite{mina21}. However, this solution relies only on \ac{chim} and does not take advantage of signals transmitted by other systems or over different frequency bands.

A second issue is that a receiver using only the authenticated signals to derive the \ac{pvt} may reduce the solution {\em accuracy} and {\em availability}. 
To make the \ac{pvt} solution more accurate,  \ac{gnss} receivers typically exploit the signals of all the available frequency bands. 
Thus, a recent trend considers cross-authentication mechanisms, where, on the receiver side, authenticated signals are combined with open non-authenticated signals, and a consistency check is performed on the obtained solution. 

To improve the trustfulness of the \ac{pvt}, \cite{navitec22} discusses a stepwise approach wherein the receiver exploits the authenticated signals as trust anchors for the complete \ac{pvt} solution. In particular, the check between signals on the same band but different carrier frequencies has been included in~\cite{cas22} to verify the consistency between Galileo encrypted E6C and E1B open measurements.  Still, the effectiveness of the proposed approach degrades when including signals from different systems and in urban scenarios.

In \cite{motellaCrossCheck}, the authors proposed a cross-authentication scheme with the goal of enhancing the Galileo ranging level protection, without relying on \ac{cas}, by leveraging on both \ac{chim} and \ac{osnma} single system authentication concepts.  Knowing that the GPS and the Galileo reference times are related by the \ac{ggto} \cite{estGGTO}, they proposed a consistency check between the measured local clock biases and the \ac{ggto} retrieved from the navigation messages. However, we will show that it is possible to attack such a mechanism even by tampering with only the open Galileo signals.

In this paper, we investigate the security of cross-authentication for \ac{pvt}-based checks in \ac{gnss}: in particular, we derive spoofing attack strategies targeting cross-authentication schemes showing that such mechanisms do not guarantee the authenticity of the computed solutions.
More in detail, the main contributions of the paper are the following.
\begin{enumerate}
    \item We introduce a new general framework for the \ac{pvt} cross-authentication checks. In particular, we identify two classes of cross-checks: time-based checks for positioning services and position-based checks for timing. The discussed state-of-the-art mechanisms fit these classes.
    \item We introduce the attack strategies that break the cross-authentication checks, and a novel generation time spoofing attack of the signal generation type.
    \item We analytically investigate the limits of the new attack, for instance, relating the degrees of the freedom of the attacker in manipulating the victim's solution to parameters of the transmission/propagation scenario.
    \item We analytically evaluate the attack performance considering both a noiseless scenario and a realistic noisy scenario.
    \item We test the attacks by using an experimental dataset collected by a Septentrio PolarRx5 receiver.
\end{enumerate}

The rest of this paper is organized as follows. Section~\ref{sec:systemmodel} introduces the system model. Section~\ref{sec:verModel} provides a general model for the \ac{pvt}-based authentication checks. Then, Section~\ref{sec:attack} describes the proposed attack strategies. Section~\ref{sec:results} presents the numerical results, obtained from the experimental dataset. Finally, Section~\ref{sec:conclusions} provides the conclusions.

\section{System Model}\label{sec:systemmodel}
We consider a multi-constellation receiver acquiring signals from $M>1$ constellations. A subset of the received signals is authenticated, i.e., protected by a \ac{sce} (or \ac{sca}) mechanism. For instance, these may represent a set of Galileo signals protected by \ac{cas} in the E6C band or GPS signals protected by \ac{chim} in L1 C/A. We assume no fault can happen in the authentication assessment of these signals; thus, the attacker cannot generate new signals to replace the authenticated ones.

In general, the victim receiver obtains signals from $N$ visible satellites: $N_\mathrm{A}$ are authenticated while the remaining $N_\mathrm{O}= N - N_\mathrm{A}$ are open (i.e., non-authenticated).
For ease of notation, the authenticated satellites are indexed by $j=1,\hdots, N_\mathrm{A}$, while non-authenticated satellites are indexed by $j=N_\mathrm{A}+1,\hdots, N$. After the general acquisition and tracking procedures (see, e.g., \cite[Ch. 5]{kaplan}), the receiver obtains a vector of pseudoranges measurements $\bm{r} = \left[ r_1, \ldots, r_\mathrm{N}\right]^\mathrm{T}$.

In nominal conditions, the \ac{pvt} solution computed by the victim receiver is 
\begin{equation}\label{eq:pvt1}
    p= [x,\, y, \, z,\, t_1, \,\ldots,\, t_{M}]^\mathrm{T}\,,
\end{equation}
denoting the receiver's three-dimensional position in \ac{ecef} coordinates, and the clock biases, one per each time reference. The procedure to obtain~\eqref{eq:pvt1} is outlined in Appendix~\ref{app:pvt1}.
Indeed, we may consider a single time reference, e.g., one of the \ac{gnss} time references or the \ac{utc}, and relate the others reference, using the respective \ac{isb} \cite[Ch. 21]{teunissen2017}, thus obtaining the \ac{pvt} solution
\begin{equation}\label{eq:pvt2}
    p= [x,\, y, \, z,\, t]^\mathrm{T}\,,
\end{equation}
as outlined in Appendix~\ref{app:pvt2}.
However, some of the considered mechanisms use instead the \ac{isb} to assess the authenticity of the solution, so, in these cases, we will estimate each clock bias separately from the others, as in~\eqref{eq:pvt1}. On the other hand, when the  assessment is based on the receiver position, we will consider the \ac{pvt} solution to be computed as in~\eqref{eq:pvt2}. 


\subsection{PVT Solution}\label{sec:pvtcomp} 
In this Section, we briefly describe how \ac{gnss} receivers typically compute the \ac{pvt} solution that will be used in the cross-authentication checks.

At each epoch, the \ac{pvt} solution is computed following an iterative approach, starting from the initial solution $\hat{\bm{p}}~=~[\hat x, \hat y, \hat z, \hat t_1, \ldots, \hat t_{M}]^\mathrm{T}$, or $\hat{\bm{p}}~=~[\hat x, \hat y, \hat z, \hat t]^\mathrm{T}$, depending on whether we aim to get~\eqref{eq:pvt1} or~\eqref{eq:pvt2}, respectively. In particular, $\hat{\bm{p}}$ is the \ac{pvt} computed at the previous epoch or a predefined starting point ({\em cold start}).
The approach follows a linearization procedure, where the \ac{pvt} is related to the range by the geometry matrix $\bm{G}$ (see Appendices~\ref{app:pvt1} and \ref{app:pvt2} for details).
Thus, given the range measurement of the $j$th signal $r_j$ and its estimate $\hat{r}_j$, the {\em range residual} is 
\begin{equation}\label{eq:rangeResiduals}
     \Delta r_j =  r_j -\hat r_j\,.
\end{equation}
Next, observe that the pseudorange displacement $\Delta\bm{r}$ is related to the position displacement $\Delta\bm{p}$ as
\begin{equation}\label{eqn:rGp}
    \Delta\bm{r} = \bm{G}\Delta\bm{p} = \begin{bmatrix}
    \bm{G}_\mathrm{A}\\
    \bm{G}_\mathrm{O}
    \end{bmatrix}\Delta\bm{p}\,,
\end{equation}
where we decomposed the geometry matrix considering the authenticated and the open signal in $\bm{G}_\mathrm{A}$ and $\bm{G}_\mathrm{O}$, respectively. From the previous relation, it follows that $\Delta\bm{p}$ can be computed from $ \Delta\bm{r}$ as 
\begin{equation}\label{eqn:pHr}
    \Delta\bm{p} = (\bm{G}^\mathrm{T}\bm{G})^{-1}\bm{G}^\mathrm{T}\Delta\bm{r} \triangleq \bm{H}\Delta\bm{r}\,,
\end{equation}
where $\bm{H}$ is the {\em Moore-Penrose pseudoinverse} of $\bm{G}$, often called also {\em least-squares solution} matrix.
We remark that the components of matrix $\bm{H}$ depend only on the relative geometry between the user and the $N$ satellites participating in the least-squares computation. 
Finally the position is updated as $\hat{\bm{p}}' = \hat{\bm{p}} +  \Delta\bm{p}$, where $\hat{\bm{p}}'$ is the starting point for the next iteration.

The iterative procedure continues until either i) the maximum number of iterations is reached or ii) the solution update is smaller than a user-defined step, i.e., $\|\Delta\bm{p}\|<\epsilon$.
More details about the procedure and the derivation of each term are reported in the Appendices. 

\subsection{Attacker Model} \label{sec:attackerModel}
In this Section, we detail the attacker's capabilities and operations. We assume that the message data (e.g., atmospheric corrections and ephemeris) are retrieved from a publicly-available authenticated side channel. Additionally, we suppose that the actual victim receiver position is known by the attacker.
We take into consideration two types of attacks:
\begin{description}
\item[\textbf{Generation attack}:] This attack consists of the generation and transmission of fake \ac{gnss} signals, corresponding possibly different \ac{pvt} solutions than the one computed using the legitimate signals. 
\item[\textbf{Relay attack}:] This attack consists of sampling, recording, and relaying the entire modulated signal with unchanged code and data. This attack is sometimes referred to as meaconing.
\end{description}

We remark that the attacker cannot perform a {generation attack} on authenticated signals, which are protected by \ac{sca} (or \ac{sce}), because he has no knowledge of the actual spreading code. 
The only action the attacker can perform on the authenticated signals is a relay attack. 
On the other hand, the attacker can generate, relay, or selectively delay the open signals, since, in this case, both the spreading code and the message content are entirely public.

We assume that the victim under attack acquires and locks onto the forged signals transmitted by the attacker. This is achieved, for instance, by increasing the transmitting power of the forged signals or canceling the legitimate open ones.

The attacker aims at inducing the victim receiver to obtain as \ac{pvt} some target position $\bm{p}^\star$ instead of the authentic solution $\bm{p}$. 
To achieve this goal, he tampers the signals propagation times, altering the range residuals $\Delta{\bm{r}}$ by a term $\Delta{\bm{r}_\mathrm{T}}$, such that, when under attack, \eqref{eqn:rGp} and~\eqref{eqn:pHr} become respectively
\begin{align}\label{eq:range_var_under_attack}
     \Delta\widetilde{\bm{r}} \triangleq \Delta\bm{r}+\Delta\bm{r}_\mathrm{T} &=\bm{G}(\Delta\bm{p}+\Delta\bm{p}_\mathrm{T}) \,, \\ 
    \Delta\widetilde{\bm{p}} \triangleq \Delta\bm{p}+\Delta\bm{p}_\mathrm{T} &=\bm{H}\Delta\widetilde{\bm{r}} \,. 
    \label{eqn:deltaptilde}
\end{align}
This means that the range tampering $\Delta\bm{r}_\mathrm{T}$ causes a displacement on the \ac{pvt} solution $\Delta\bm{p}_\mathrm{T}$. So, the \ac{pvt} update due to the attacker spoofing will be 
\begin{equation}\label{eq:shift_att}
    \widetilde{\bm{p}} =  \hat{\bm{p}} + \Delta\widetilde{\bm{p}} =  \hat{\bm{p}} + \Delta\bm{p}+\Delta\bm{p}_\mathrm{T}\, .
\end{equation}
We remark that the iterative least squares approach to compute the \ac{pvt} when under attack is approximated by a linear transformation only when the pseudorange displacement $\Delta\bm{r}_\mathrm{T}$ introduced by the attacker is small with respect to $\bm{r}$.
In Section~\ref{sec:results}, we will test the effectiveness of the attack by exploiting the linearization procedure and evaluating the performance as a function of the distance between the target and the real position, i.e., the magnitude of the pseudorange displacements to be induced, thus testing the robustness of the linear approximation.

\begin{figure}
    \centering
    \tikzset{every picture/.style={line width=0.75pt}} 

\begin{tikzpicture}[x=0.75pt,y=0.75pt,yscale=-1,xscale=1]
\tikzstyle{every node}=[font=\small]

\draw [color={rgb, 255:red, 208; green, 2; blue, 27 }  ,draw opacity=1 ]   (37.66,36.54) -- (196.15,28.63) ;
\draw [shift={(34.67,36.69)}, rotate = 357.14] [fill={rgb, 255:red, 208; green, 2; blue, 27 }  ,fill opacity=1 ][line width=0.08]  [draw opacity=0] (3.57,-1.72) -- (0,0) -- (3.57,1.72) -- cycle    ;
\draw    (204.47,87.48) -- (193.19,28.63) ;
\draw [shift={(205.04,90.42)}, rotate = 259.14] [fill={rgb, 255:red, 0; green, 0; blue, 0 }  ][line width=0.08]  [draw opacity=0] (3.57,-1.72) -- (0,0) -- (3.57,1.72) -- cycle    ;
\draw [color={rgb, 255:red, 208; green, 2; blue, 27 }  ,draw opacity=1 ] [dash pattern={on 4.5pt off 4.5pt}]  (196.15,28.63) -- (59.85,62.21) ;
\draw    (64.27,63.5) -- (205.04,93.11) ;
\draw [shift={(61.33,62.89)}, rotate = 11.88] [fill={rgb, 255:red, 0; green, 0; blue, 0 }  ][line width=0.08]  [draw opacity=0] (3.57,-1.72) -- (0,0) -- (3.57,1.72) -- cycle    ;
\draw  [color={rgb, 255:red, 208; green, 2; blue, 27 }  ,draw opacity=1 ][fill={rgb, 255:red, 208; green, 2; blue, 27 }  ,fill opacity=1 ] (29.93,36.69) .. controls (29.93,35.5) and (30.99,34.54) .. (32.3,34.54) .. controls (33.61,34.54) and (34.67,35.5) .. (34.67,36.69) .. controls (34.67,37.88) and (33.61,38.84) .. (32.3,38.84) .. controls (30.99,38.84) and (29.93,37.88) .. (29.93,36.69) -- cycle ;
\draw  [fill={rgb, 255:red, 255; green, 255; blue, 255 }  ,fill opacity=1 ] (190.22,28.63) .. controls (190.22,27.15) and (191.55,25.94) .. (193.19,25.94) .. controls (194.82,25.94) and (196.15,27.15) .. (196.15,28.63) .. controls (196.15,30.11) and (194.82,31.32) .. (193.19,31.32) .. controls (191.55,31.32) and (190.22,30.11) .. (190.22,28.63) -- cycle ;
\draw  [fill={rgb, 255:red, 0; green, 0; blue, 0 }  ,fill opacity=1 ] (202.07,93.11) .. controls (202.07,91.63) and (203.4,90.42) .. (205.04,90.42) .. controls (206.67,90.42) and (208,91.63) .. (208,93.11) .. controls (208,94.6) and (206.67,95.8) .. (205.04,95.8) .. controls (203.4,95.8) and (202.07,94.6) .. (202.07,93.11) -- cycle ;
\draw  [fill={rgb, 255:red, 208; green, 2; blue, 27 }  ,fill opacity=1 ] (56.89,62.21) .. controls (56.89,60.73) and (58.22,59.53) .. (59.85,59.53) .. controls (61.49,59.53) and (62.81,60.73) .. (62.81,62.21) .. controls (62.81,63.7) and (61.49,64.9) .. (59.85,64.9) .. controls (58.22,64.9) and (56.89,63.7) .. (56.89,62.21) -- cycle ;

\draw (92,15) node [anchor=north west][inner sep=0.75pt]  [color={rgb, 255:red, 208; green, 2; blue, 27 }  ,opacity=1 ] [align=left] {$\Delta\bm{p}^\star$};
\draw (21,18) node [anchor=north west][inner sep=0.75pt]  [color={rgb, 255:red, 208; green, 2; blue, 27 }  ,opacity=1 ] [align=left] {$\bm{p}^\star$};
\draw (200,15) node [anchor=north west][inner sep=0.75pt]   [align=left] {$\hat{\bm{p}}$};
\draw (212,88) node [anchor=north west][inner sep=0.75pt]   [align=left] {$\bm{p}$};
\draw (90,75) node [anchor=north west][inner sep=0.75pt]   [align=left] {$\Delta\bm{p}_\mathrm{T}$};
\draw (135,42) node [anchor=north west][inner sep=0.75pt]   [align=left] {$\Delta\widetilde{\bm{p}}$};
\draw (200,50) node [anchor=north west][inner sep=0.75pt]   [align=left] {$\Delta{\bm{p}}$};
\draw (40,60) node [anchor=north west][inner sep=0.75pt]  [color={rgb, 255:red, 208; green, 2; blue, 27 }  ,opacity=1 ] [align=left] {$\Tilde{\bm{p}}$};

\end{tikzpicture}
    \caption{Example of the scenario under attack: $\hat{\bm{p}}$ is the solution at the iteration start, $\bm{p}$ is the legitimate solution, ${\bm{p}}^\star$ is the attacker target solution, and $\tilde{\bm{p}}$ is the solution computed by the victim, mixing open and tampered signals.}
    \label{fig:distances}
\end{figure}
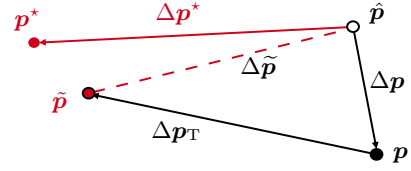

Fig.~\ref{fig:distances} shows the displacement of the attacker \ac{pvt} due to the attacker tampering: in legitimate conditions, starting from $\hat{\bm{p}}$, the receiver would compute $\Delta \bm{p}$ to obtain the legitimate solution $\bm{p}$ (in black). In turn, the attacker wants to lead the receiver to the target $\bm{p}^\star$ (in red), thus he tampers the signals, causing the alteration $\Delta\bm{p}_\mathrm{T}$. In this case, instead of the true \ac{pvt}, the victim computes $\tilde{\bm{p}}$, as close as possible to the target $\bm{p}^\star$.
Notice that Fig.~\ref{fig:distances} may represent the actual 3D position or the whole \ac{pvt} vector, i.e., considering points in the $3 + M$ space.

For the sake of clarity, we will distinguish between authenticated and non-authenticated ranges residuals, $ \Delta\bm{r}_\mathrm{T,A}$ and $\Delta\bm{r}_\mathrm{T,O}$, that the attacker introduces to the authenticated and non authenticated signals, respectively 
\begin{equation}
    \Delta\bm{r}_\mathrm{T} =
    \begin{bmatrix}
    \Delta\bm{r}_\mathrm{T,A}\\
    \Delta\bm{r}_\mathrm{T,O}
    \end{bmatrix}.
\end{equation}
Finally, it is natural to partition the least square matrix as $\bm{H} = [\bm{H}_\mathrm{A} \, \bm{H}_\mathrm{O} ]$, where $\bm{H}_\mathrm{A}$ and $\bm{H}_\mathrm{O}$ represent the matrix associated to the authenticated and the open signals, respectively. In particular, since the attacker can only perform a relay attack on the authenticated signals, we model the tampering of the attacker on the authenticated ranges as $\Delta\bm{r}_\mathrm{T,A} = k\bm{1}$, where $\bm{1}$ is the column vector filled with all ones. Indeed, when performing a generation attack, the attacker sets $k =0$, i.e., only the open signals are actually fake.


\section{PVT-based Cross-Authentication Checks Models}\label{sec:verModel} 
The receiver computes the \ac{pvt} solution using both authenticated and non-authenticated ranges. Next, to assess its authenticity, it performs one or more consistency checks on the \ac{pvt} (or in part of it). 
We use decision theory to describe the problem and consider the two hypotheses 
\begin{align*}
    &\mathcal{H}_0 :\;   \mbox{PVT is legitimate, signals are legitimate,}\\
    &\mathcal{H}_1 :\;   \mbox{PVT is fake, signals have been tampered with.}
\end{align*}
The consistency check on the \ac{pvt} provides one of the two decision 
\begin{align*}
    &\hat{\mathcal{H}}_0 :\;   \mbox{Check passed, PVT is legitimate}\\
    &\hat{\mathcal{H}}_1 :\;   \mbox{Check failed, PVT is not legitimate.}
\end{align*}

This allows us to define the {\em false alarm} and {\em missed detection} probabilities as 
\begin{subequations}
    \begin{equation}
        p_\mathrm{FA} = P(\hat{\mathcal{H}}_1 | {\mathcal{H}}_0)\,
    \end{equation}
    \begin{equation}
        p_\mathrm{MD} = P(\hat{\mathcal{H}}_0 | {\mathcal{H}}_1)\,.
    \end{equation}
\end{subequations}   

In the next Sections, we describe the classes of consistency checks to authenticate the \ac{pvt}, also related to the required service.

\subsection{Time-Based Checks for Navigation Services}\label{sec:TimingNavi}
    \begin{figure}
        \centering
        \includegraphics[width=\columnwidth]{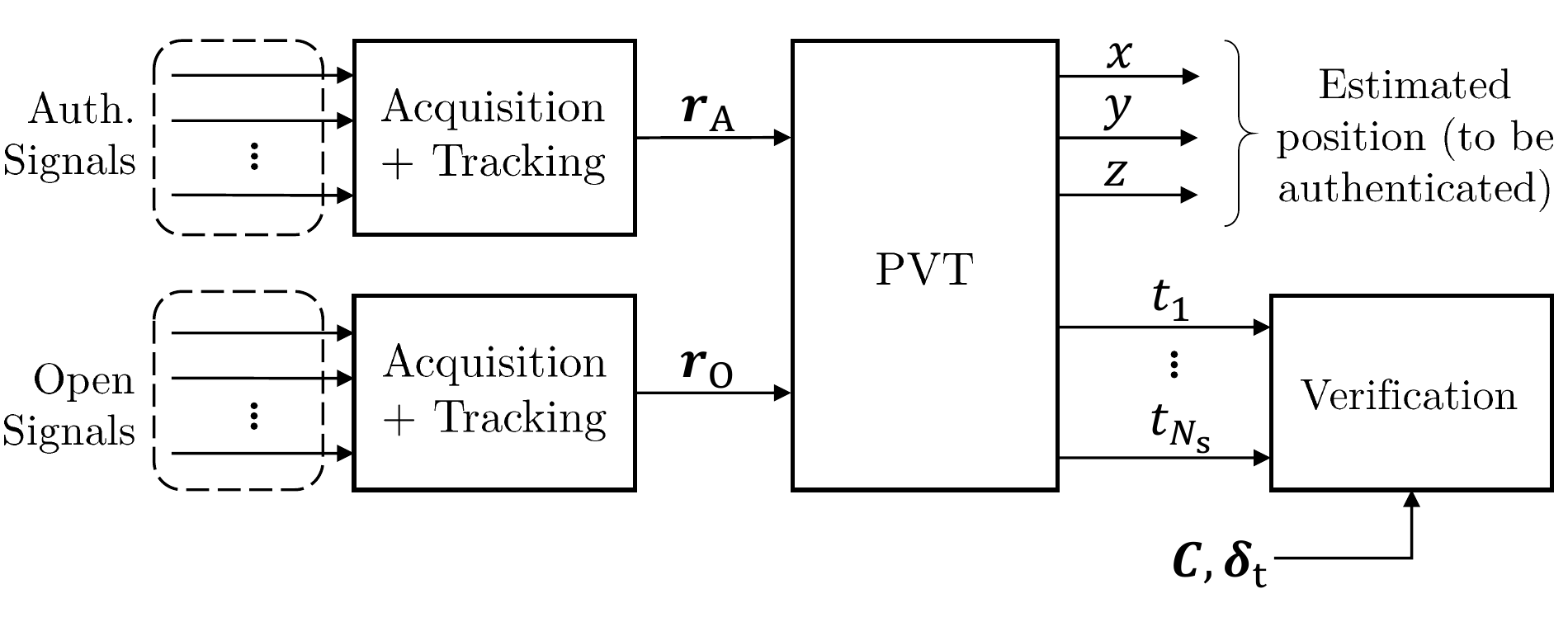}
        \caption{Block scheme of the time-based check.}
        \label{fig:time_check}
    \end{figure}
The receiver collects signals from $M$ different \acp{gnss}, with the \ac{pvt} solution containing $N_{\rm s}$ clock biases, one per \ac{gnss} time reference, thus, $\bm{p}$ is computed as in Appendix~\ref{app:pvt1}. The time references are related by the \ac{isb}, thus the receiver considers the \ac{pvt} to be authentic if 
\begin{equation}\label{eq:generalCheck}
    \bm{\theta}_\mathrm{t} \triangleq \bm{C} \bm{p}   \leq \bm{\delta}_\mathrm{t}\,,
\end{equation}
    where $\bm{C}$ is a matrix that selects and relates a pair of clock biases (with the appropriate sign), while $\bm{\delta}_\mathrm{t}$ contains an upper bound of the expected \ac{isb}, eventually tuned to meet a predefined false alarm. The mechanism for time-based checks is summarized in Fig.~\ref{fig:time_check}.
    Matrix $\bm{C}$ has $N_\mathrm{c} = 2 \left( M -1\right)$ rows, with $\bm{\delta}_\mathrm{t}$ containing $M -1$  upper-bounds on the measured \ac{isb}. Therefore, $\bm{C}$ has $M-1$ linearly independent rows. The $i$-th row of $\bm C$ is $\bm{c}_i$: this will be associated to the $i$-th element of $\bm \delta_\mathrm{t}$, $\delta_{\mathrm{t},i}$, such that the $i$-th check will be $\theta_{\mathrm{t},i} = \bm{c}_i \bm{H} \bm{r} \leq\delta_{\mathrm{t},i} $. For instance, in~\cite{motellaCrossCheck} authors propose to authenticate the computed \ac{pvt} solution through a time consistency check. In particular, they consider the \ac{pvt} solution obtained by combining the signals from two constellations, one transmitting authenticated signals and the other transmitting open signals, following the procedure outlined in Appendix~\ref{app:pvt1}. Then, denoting as $t_\mathrm{1}$ and $t_\mathrm{2}$ the clock biases associated with the time reference of each constellation, the \ac{pvt} is authenticated by checking, at the receiver, if the difference $|t_1-t_2|$ is consistent with the reference \ac{isb}, as
\begin{equation}\label{eqn:test}
    \left|(t_\mathrm{2}-t_\mathrm{1}) - b-\mathrm{ISB}\right|<T\;,
\end{equation}
where $T$ is a threshold that is set by the receiver, and $b$ is a bias set during the calibration phase. More details about the security check and the parameters' derivation are reported in~\cite{motellaCrossCheck}. 
The check~\eqref{eqn:test} can be framed in~\eqref{eq:generalCheck}, by choosing 
\begin{equation}\label{eq:modelMotella}
    \bm{C} =     
    \begin{bmatrix}
     0 & 0 &  0 &  1 & -1 \\
     0 & 0 & 0  & -1 & 1 \\
    \end{bmatrix}\,, 
    \quad \bm{\delta}_\mathrm{t} = 
    \begin{bmatrix}
        T + b +\mathrm{ISB} \\
        T - b-\mathrm{ISB}        
    \end{bmatrix}.
    \end{equation}
This means that the attacker must craft a $\Delta\bm{r}_\mathrm{T}$ that changes the estimated time offset by at most $\bm{\delta_\mathrm{t}}$. In the rest of the paper, we will refer to this scheme as \textit{\ac{isb}-based cross-authentication check}.

\subsection{Position-Based Checks for Timing Services}\label{sec:positionChecks}
    \begin{figure}
        \centering
        \includegraphics[width=\columnwidth]{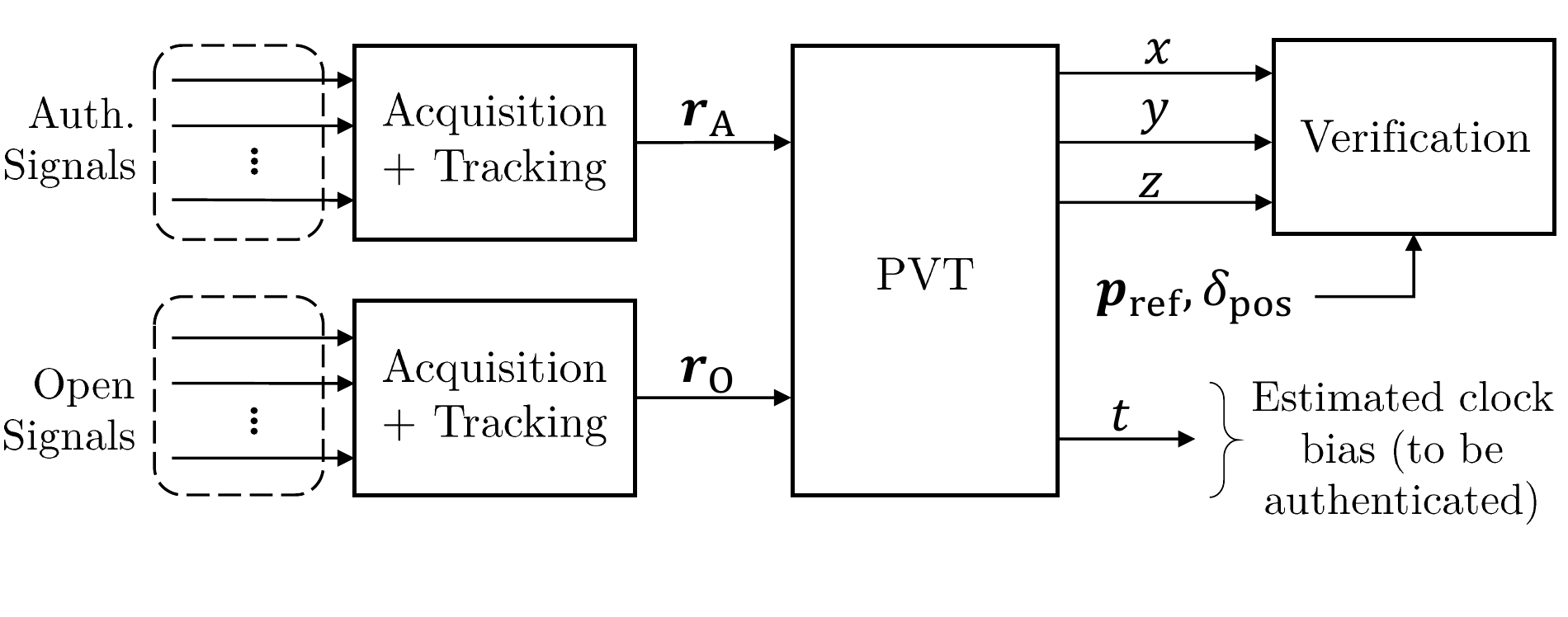}
        \caption{Block scheme of the position-based check.}
        \label{fig:pos_check}
    \end{figure}
The goal here is to provide a secure \ac{gnss}-based timing source. We assume the receiver knows a priory its position, whose coordinates are denoted by $\bm{p}_\mathrm{ref}=[x_\mathrm{ref}, y_\mathrm{ref}, z_\mathrm{ref}]$: these may be obtained either by computing the \ac{pvt} using only authenticated ranges, received by a side channel, or it can be simply a priory (e.g., in a static scenario). However, none of these side channels provides a reliable timing correction. Hence, the receiver exploits all the available measurements, both authenticated and open, and obtains the \ac{pvt} solution $\bm{p} =[x,y,z,t]$, containing the desired clock bias correction, as outlined in Appendix~\ref{app:pvt2}. To verify its authenticity, the receiver verifies that the reconstructed position is consistent with its a-priori knowledge of the position, verifying that 
\begin{multline}\label{eq:testTiming}
{\theta}_\mathrm{pos} \triangleq \sqrt{\left(x_\mathrm{ref} - x \right)^2 + \left(y_\mathrm{ref} - y \right)^2 +\left(z_\mathrm{ref} - z \right)^2 }\\ = \sqrt{\bm{\epsilon^\mathrm{T} \epsilon}}
 \leq \delta_\mathrm{pos}  \,.
\end{multline}
where $\bm \epsilon \triangleq [x_\mathrm{ref}-x, y_\mathrm{ref}-y, z_\mathrm{ref}-z]$ represents the 3D position displacement with respect to the reference position.
The position-based check is summarized in Fig.~\ref{fig:pos_check}.
If the check passes, i.e., the position computed using all the signals is coherent with the reference one, the solution is considered safe and the receiver uses the timing correction from $\bm{p}$ to correct the local clock. An example of such a setting is described in~\cite{paperSensors}.

\section{Proposed Attack Strategies}\label{sec:attack}
In this Section, we outline the proposed strategies that the attacker uses to alter the \ac{pvt} solution without alerting the victim receiver. First, we focus on the ideal noiseless scenario, where the received signals of neither the attacker nor the victim are affected by noise. We will generalize the attack for a more realistic scenario in Section~\ref{sec:noise}.

To be successful the attacker must i) alter only the non-protected signals, and choose the $\Delta \bm{r}_\mathrm{T}$ such that ii) $\Delta\bm{ \tilde{p}}$ passes the check employed by the victim, and iii) leads the victim to a \ac{pvt} solution close to the target $\bm{p}^\star$.
In this first phase, we assume that the victim receives the signals as intended by the attacker. Moreover, we consider that the attacker knows the pseudorange vector $\bm{r}$ that the victim would measure in the nominal scenario.
In Section~\ref{sec:noise}, we will analyze the scenario wherein the attacker cannot pre-compensate the channel, and the victim receives a noisier copy of the tampered signal.

In the following Sections, we describe the attacking strategies for the timing and position-based checks.

\subsection{Generation Attack Against Time-Based Checks}\label{sec:genatt_VS_TimeCheck}
We consider a legitimate receiver employing the check described in Section~\ref{sec:TimingNavi}, where the \ac{pvt} solution is considered to be authentic and the clock bias differences are coherent with the \acp{isb} derived from a side channel.

First, we identify the set of {\em feasible solutions} that do not raise an alarm at the receiver. Next, we pick the solution that leads the receiver estimated position close to the target position $\bm{p}^\star$.

To identify the set of feasible solutions, we have the following result.
\begin{theorem} Given a constraint vector $\bm{\delta}'$, the set of feasible solutions for the attacker containing all the range alterations that do not raise an alarm in~\eqref{eq:generalCheck} is the affine space
\begin{equation}\label{eq:ammissibleSolution}
    S =  \left\{ \Delta\bm{r}_\mathrm{p} + s\,,  s \in \mathcal{N}(\bm{CH}_\mathrm{O}) \right\} \,,
\end{equation}
where $\Delta\bm{r}_\mathrm{p}$ is a solution of $\bm{\delta}' = \bm{CH}_\mathrm{O}\Delta\bm{r}_\mathrm{T,O}$, and $\mathcal{N}(\bm{X})$ is the null space of $\bm{X}$.
Moreover, it holds 
\begin{equation}\label{eq:spaceSize}
     \mathrm{dim}({S}) =  N_\mathrm{O} -  M +1\;.
\end{equation}
\end{theorem}
\begin{proof}

Now, condition~\eqref{eq:generalCheck} requires
\begin{equation}\label{eq:Cpdelta}
    \bm{C}\widetilde{\bm{p}} = \bm{C\hat{p}}+\bm{C} \Delta\widetilde{\bm{p}} = \bm{C\hat{p}}+\bm{C}\Delta\bm{p}+ \bm{C}\Delta\bm{p}_\mathrm{T}  \leq \bm{\delta}_\mathrm{t} \,.
\end{equation}
We assume that, in the nominal scenario, the check~\eqref{eq:generalCheck} is passed, thus, $ \bm{C\hat{p}}+\bm{C}\Delta\bm{p} = \bm{\delta}'' \leq \bm{\delta}_\mathrm{t}$: for instance, this models also the case where the defender's measurements are affected by noise. So, to avoid raising an alarm, the attacker can choose any alteration that induce a change smaller than $\bm{\delta}_\mathrm{t} -\bm{\delta}''$, i.e., from \ref{eq:Cpdelta} the shift induced by the attacker is 
\begin{equation}\label{eq:Cpdelta2}
     \bm{C}\Delta\bm{p}_\mathrm{T}  \triangleq \bm{\delta}'\,,\; \; \bm{\delta}' \leq\bm{\delta}_\mathrm{t} -\bm{\delta}''\,.
\end{equation}

By using~\eqref{eqn:pHr} and considering that only non-protected signals can be manipulated by the attacker, we have 
\begin{equation}\begin{split}  \label{eq:cond_delta'}
         \bm{\delta}' = \bm{C} \Delta\bm{p}_\mathrm{T}  =  \bm{C} \bm{H} \Delta \bm{r}_\mathrm{T} = \bm{C} \begin{bmatrix}
            \bm{H}_\mathrm{A} & \bm{H}_\mathrm{O} 
        \end{bmatrix} 
    \begin{bmatrix}
            \bm{0}\\
            \Delta\bm{r}_\mathrm{T,O}
    \end{bmatrix} \\ = 
         \bm{C} \bm{H}_\mathrm{O}  \Delta\bm{r}_\mathrm{T,O}\;.
\end{split}
\end{equation}
Being $\bm{\delta}'$ under the control of the attacker, he can choose it such that $\bm{\delta}'\in \mathcal{R}(\bm{CH}_\mathrm{O})$, so that~\eqref{eq:cond_delta'} has at least one solution.
Indeed, $\mathrm{rank}(\bm{CH}_\mathrm{O}) = M-1$, since $\bm{C}$ has $ M-1$ linearly independent rows. \footnote{We made the tacit assumption that $N_\mathrm{O}\geq M$ and that the geometry are such that $\bm{C}\bm{H}_\mathrm{O}$ is full row rank.}
The matrix resulting from the product $\bm{CH}_\mathrm{O}$ has size $N_\mathrm{c}\times N_\mathrm{O}$. 

Hence, by the Rouché-Capelli's theorem,
\eqref{eq:cond_delta'} describes a linear undetermined system that admits infinitely many solutions, belonging to the affine space $S$ as defined in~\eqref{eq:ammissibleSolution}, where $\Delta\bm{r}_\mathrm{p}$ is any particular solution of \eqref{eq:cond_delta'}, which can be found, for instance, by Gaussian elimination. 
Moreover, since the dimension of an affine space is the dimension of its associated linear space, we have
\begin{multline}
    \mathrm{dim}(S) = \mathrm{dim}(\mathcal{N}(\bm{CH}_\mathrm{O}))\\ 
    = N_\mathrm{O} - \mathrm{rank}(\bm{CH}_\mathrm{O}) =  N_\mathrm{O} -  M +1\, > 1\,.
\end{multline}
So, each solution of \ref{eq:cond_delta'} depends on $N_\mathrm{O} -  M +1$ parameters.
\end{proof}
The above Theorem has two main consequences.
First, let us consider the worst-case scenario for the attacker, where the legitimate range residuals $\Delta\bm{r}$ already yield $\bm{C}\Delta\bm{p} = \bm{\delta}$. In this case, there is no margin for the attacker that, from \ref{eq:Cpdelta2}, has to pick an attack with $\bm{\delta}' = 0$. Still, \eqref{eq:ammissibleSolution} shows that it is possible to lead a successful attack, by picking ranges from the null space $\mathcal{N}(\bm{CH}_\mathrm{O})$. 

Secondly, it shows that the increased number of signals used in the \ac{pvt} by the legitimate receiver potentially leads to an increased degree of freedom given to the attacker to manipulate the fake \ac{pvt} solution, since the dimension of the null space increases.
However, using more open signals improves the accuracy of the \ac{pvt} estimation, when under attack. Therefore, a trade-off between accuracy and security is at stake here.
Finally, there may be scenarios where $N_\mathrm{A}<4$, therefore no fully authenticated \ac{pvt} solution is available at all. In these cases, the victim receiver has to increase the number of the open signals used for the \ac{pvt} or not compute any \ac{pvt} at all.



Next, we derive the actual range alterations that lead the victim receiver to the target position $\bm{p}^\star$.
We consider the worst-case scenario for the attacker, where $\bm{\delta}' = \bm{0}$. In this case, \eqref{eq:cond_delta'} describes a linear homogeneous system whose solutions belong to the linear space $S = \mathcal{N}(\bm{CH_\mathrm{O}})$.
Indeed, this boils down to the general case by adding in the calculations a range vector related to the particular solution $\Delta \bm{r}_\mathrm{p}$.
Thus, any solution belonging to $S$ can be written as 
\begin{equation}\label{eq:subspaceKernel}
    \Delta \bm{r}_\mathrm{T} = 
    \begin{bmatrix}
        \bm{u}_1 & \ldots  &  \bm{u}_K \\
    \end{bmatrix} 
    \begin{bmatrix}
        \alpha_1 \\ \ldots  \\  \alpha_K \\
    \end{bmatrix} = \bm{U}\bm{\alpha} \;,
\end{equation}
where $\bm{u}_1,\ldots, \bm{u}_K$ is an orthonormal basis of $S$, $\alpha_k \in \mathbb{R}$ $\forall k  = 1,\ldots, K$, and $K \triangleq \mathrm{dim}\left(S\right)$. Notice that $\bm{U}$ is a matrix with $3 + M$ rows and $K$ columns.

Target and fake positions are related by
\begin{multline}\label{eq:distAlter}
    \widetilde{\bm{p}}-\bm{p}^\star =  \Delta\widetilde{\bm{p}}-\Delta\bm{p}^\star =  \bm{H} \Delta \bm{r}_\mathrm{T} - \left(\Delta\bm{p}^\star - \Delta\bm{p}\right)   \\= \bm{H} \bm{U}\bm{\alpha} -  \left(\Delta\bm{p}^\star - \Delta\bm{p}\right)\;,
\end{multline}
where in the last step we used the fact that the attacker will only use the range alterations that do not raise an alarm, i.e., belonging to~\eqref{eq:ammissibleSolution} so having the form in~\eqref{eq:distAlter}.

By imposing the target position to coincide to the fake one, i.e., $\widetilde{\bm{p}}=\bm{p}^\star$, we have 
\begin{equation}\label{eq:attackOptim}
    \bm{H} \bm{U} \bm{\alpha} =  \Delta\bm{p}^\star - \Delta\bm{p},
\end{equation}
and, resorting to the Moore-Penrose pseudoinverse, we obtain 
\begin{equation}
   \bm{\alpha}^\star = \left( \bm{U} ^\mathrm{T} \bm{H}^\mathrm{T}  \bm{H} \bm{U} \right)^{-1}  \bm{U} ^\mathrm{T} \bm{H}^\mathrm{T}  \left(\Delta\bm{p}^\star - \Delta\bm{p}\right)\,.
\end{equation}
Finally, plugging this in~\eqref{eq:subspaceKernel}, we obtain the target range alteration 
\begin{equation}\label{eq:solAlgebra}
     \Delta\bm{r}_\mathrm{T}^\star = \alpha^\star_1 \bm{u}_1 + \ldots +  \alpha^\star_K \bm{u}_K \;.
\end{equation}
Indeed, as required, this attack will not raise any alarm while taking the victim exactly to $\bm{p}^\star$.

Alternatively, by exploiting~\eqref{eq:distAlter}, we can choose the solution directly minimizing the distance between the induced and the target position, among those that do not raise an alarm, i.e., 
\begin{equation}\label{eq:minProb}
    \Delta\bm{r}_\mathrm{T}^\star =  \argmin_{\Delta\bm{r}_\mathrm{T}\in S} \| \bm{H} \Delta \bm{r}_\mathrm{T} - \left(\Delta\bm{p}^\star - \Delta\bm{p}\right)\|^2\,.
\end{equation}
Indeed, as long as the space $S$ is not empty, this problem can be solved by numeric methods.

\paragraph*{Attack Against the GGTO-based Cross-authentication Check}
We consider the authentication procedure in~\cite{motellaCrossCheck} and exploit the previous analysis to derive the attack. 
The technique proposed in~\cite{motellaCrossCheck} only considers signals from Galileo and GPS, so we have $M = 2$. According to~\eqref{eq:spaceSize}, we expect to find a space of feasible solutions $S$ with size $\mathrm{dim}\, (S)= N_\mathrm{O} - 1$.
We consider the worst-case scenario for the attacker where $\bm{\delta}' = \bm{0}$. 
Combining~\eqref{eq:modelMotella} and~\eqref{eq:cond_delta'} and denoting as $\bm{h}_{\mathrm{O},i}$ the $i$-th row of $\bm{H}_\mathrm{O}$, we have
\begin{equation}\label{eqn:equaltimes}
    \bm{C}\bm{H}_\mathrm{O}\Delta\bm{r}_\mathrm{T,O} = \bm{0} \Leftrightarrow
    (\bm{h}_{\mathrm{O},4}-\bm{h}_{\mathrm{O},5})\Delta\bm{r}_\mathrm{T,O}=0\,.\\
\end{equation}
Thus, $\Delta\bm{r}_\mathrm{T,O}$ belongs to the null space $\mathcal{N}(\bm{h}_{\mathrm{O},4}-\bm{h}_{\mathrm{O},5})$. This corresponds to the orthogonal complement to the space generated by column vector $(\bm{h}_{\mathrm{O},4}-\bm{h}_{\mathrm{O},5})^\mathrm{T}$, or
\begin{equation}\label{eq:setRMotella}
  S=  \mathcal{N}(\bm{h}_{\mathrm{O},4}-\bm{h}_{\mathrm{O},5}) = \left\langle(\bm{h}_{\mathrm{O},4}-\bm{h}_{\mathrm{O},5})^\mathrm{T}\right\rangle^\perp\,.
\end{equation}
Considering that $\bm{h}_{\mathrm{O},4}$ and $\bm{h}_{\mathrm{O},5}$ are column vectors of length $N_\mathrm{O}$, the null space in~\eqref{eq:setRMotella} has clearly size $N_\mathrm{O}-1$ and, as expected, $\mbox{dim}(S) =N_\mathrm{O}-1 $.

\subsection{Attack Against Position-Based Checks}\label{sec:attTimePositionCheck}
We consider the class of checks described in Section~\ref{sec:positionChecks}. Here, the victim is using \ac{gnss} for timing, while it is monitoring its own position, using the procedure outlined in Appendix~\ref{app:pvt2}, which is used for the security check: for instance, a scenario where \ac{gnss} signals are used to synchronize the clock of a \ac{lan} within one building is presented in~\cite{paperSensors}. Indeed, in such condition the position is well-known, so the attacker can only tamper the timing estimation, and induce a \ac{pvt} shift of the type 
\begin{equation}\label{eq:timingShift}
    \Delta\bm{p}_\mathrm{T} =  \begin{bmatrix}
       \xi_x \\  \xi_y \\   \xi_z \\ c\gamma_{T}
    \end{bmatrix}
    = \begin{bmatrix}
       \xi_x \\  \xi_y \\   \xi_z \\ 0
    \end{bmatrix} 
    + 
     \begin{bmatrix}
       0\\  0\\  0 \\ c\gamma_{T}
    \end{bmatrix}
    = \bm{\xi} + 
     \begin{bmatrix}
       0\\  0\\  0 \\ c\gamma_{T}
    \end{bmatrix}\,,
\end{equation}
where $c\gamma_{T}$ is the time shift induced by the attacker to the victim clock. To not raise an alarm, by~\eqref{eq:testTiming} we need $\|\bm{\xi}\|^2 \leq\delta_\mathrm{pos}$.

In the next paragraphs, we will present a relay and a generation attack against the position-based check. Using either one, the attacker is able to induce the target time shift ~\eqref{eq:timingShift}. While a generation attack may be preferable since, for instance, it does not introduce additional noise at the victim due to the attacker receiver noise, we will prove that it can only be performed if $N_\mathrm{O}\geq 4$ signals are used by the victim.

\paragraph{Relay Attack Against the Position-Based Check}
Starting from \ref{eq:timingShift} the attack can be compute straightly from \ref{eqn:rGp} as 
\begin{equation}\label{eq:meaconingMargin}
    \Delta\bm{r}_\mathrm{T} = \bm{G}  \Delta\bm{p}_\mathrm{T} = c\gamma_{T}\bm{1} + G\bm{\xi}\;,
\end{equation}
where $\bm{1}$ is a vector of ones with the same size as $\bm{p}$. 
In the worst-case scenario for the attacker, where $\| \bm{\xi} \| =0 $ or, equivalently, $\delta_\mathrm{pos} = 0$, i.e., no change in the position is tolerated at the legitimate receiver, it yields 
\begin{equation}\label{eq:meaconing}
    \Delta\bm{r}_\mathrm{T} = c\gamma_{T}\bm{1}\;.
\end{equation}
This means that the attacker has to receive, record, and retransmit with some delay both the authenticated and the open signals to induce the required attack, as typically done for the meaconing attack. 
On the other hand, this also shows that when picking $\bm{\xi} \neq \bm{0}$ we are choosing a solution that differs by (at most) $G\bm{\xi}$ from the meaconing solution. 


Finally notice that, as pointed out in Section~\ref{sec:attackerModel}, this attack can be performed even when the ranges are authenticated by \ac{sce} or \ac{sca}.
Still, this is a viable option only if the attacker is sufficiently close to the victim. Conversely, it may induce a different position calculation, that may alert the victim. 
However, especially when considering realistic scenarios, relay attacks may introduce additional noise. Thus, if possible, the attacker would rather use a generation attack even for position-based checks. For this reason, in the next paragraph, we will present a generation attack strategy against position-based checks.

\paragraph{Generation Attack Against Position-Based Checks}\label{sec:generationAttack}

In this Section, we show how and when it is possible for the attacker to perform a generation-type attack against the position-based check. The attacker aims at introducing a time shift equal to $c\gamma_\mathrm{T}$ on the estimated clock bias, as in~\eqref{eq:meaconing}.
So, by exploiting the linearization procedure of \ref{eqn:pHr} and combining it with~\eqref{eq:timingShift}, the attacker must choose $\Delta r_\mathrm{T}$ to obtain
\begin{equation}\label{eq:GenTimingProb}
    \bm{H} \Delta\bm{r}_\mathrm{T} = 
     \begin{bmatrix}
       0\\  0\\  0 \\ c\gamma_{T}
    \end{bmatrix}\,.
\end{equation} 
Indeed, the range alteration $\Delta\bm{r}_\mathrm{T}$ that solves~\eqref{eq:GenTimingProb} induces a shift only in the time component of the \ac{pvt} solution computed by the victim. 
We remark that, for ease of reading, we focus on the worst-case scenario where $\|\xi\|=0$. This can be easily extended also for the case where $\bm{\xi}\neq \bm{0}$, by adding a term $\bm{G}\bm{\xi}$ in~\eqref{eq:GenTimingProb}.

Analogously to the procedure outlined in Section~\ref{sec:genatt_VS_TimeCheck}, we have that the set of feasible solutions of \ref{eq:GenTimingProb} is the affine space 
\begin{equation}\label{eq:ammissibleSolutionTiming}
    S =  \left\{ \Delta\bm{r}_\mathrm{p} + s\,,  s \in \mathcal{N}(\bm{H}) \right\} \,.
\end{equation}
where, from \ref{eq:meaconing}, the particular solution is the meaconing attack $\Delta\bm{r}_\mathrm{p} = c\gamma_\mathrm{T}\bm{1}$.

Indeed, $\bm{H}$ is $4\times N$  and $\mathrm{rank}(\bm{H})=4$, therefore
\begin{equation}
    \mathrm{dim}(\mathcal{N}(\bm{H}) ) = \mathcal{R}(\bm{H}) - \mathrm{rank}({\bm{H}}) = N -4.
\end{equation}
We remark that to compute the \ac{pvt}, it always holds $N\geq 4$, so $\mathrm{dim}(\mathcal{N}(\bm{H}) ) \geq 0$.
Hence, any range tampering that induces the target time shift, can be written as  
\begin{equation}
    s = c\gamma_\mathrm{T}\bm{1} + \sum_{n = 1}^{N-4}\beta_n \bm{u}_n = c\gamma_\mathrm{T}\bm{1} + \bm{U}\bm{\beta}
\end{equation}
where vectors $\bm{u}_1,\ldots, \bm{u}_{N-4}$ form a basis for the null space $\mathcal{N}(\bm{H})$, collected as the columns of the basis matrix $\bm{U}$.

Our aim is to find the range tampering that belongs to $S$ and does not require any alteration on the authenticated ranges. This latter requirement is formalized as 
\begin{equation}\label{eq:linearTiming}
    \bm{U}_{N_\mathrm{A}}\bm{\beta} + c\gamma_\mathrm{T}\bm{1}_{N_\mathrm{A}} = 0\;,
\end{equation}
where $\bm{U}_{N_\mathrm{A}}$ and  $\bm{1}_{N_\mathrm{A}}$ collect the first $N_\mathrm{A}$ rows of $\bm{U}$ and  $\bm{1}$, respectively.

A sufficient condition for~\eqref{eq:linearTiming} to admit at least one solution is that $\bm{U}_{N_\mathrm{A}}$, with size $N_\mathrm{A}\times N-4$, be left invertible. This happens if $N_\mathrm{A} < (N -4)$ or, equivalently, $N_\mathrm{O}>4$.\footnote{Under the condition $N_\mathrm{O}>4$, $\bm{U}_{N_\mathrm{A}}$ is left invertible only if the basis vectors $\bm{u}_1,\ldots, \bm{u}_{N-4}$ of the null space $\mathcal{N}(\bm{H})$ are chosen in such a way that $\bm{U}_{N_\mathrm{A}}$ has full row rank.} 
Moreover this also assures that $c\gamma_\mathrm{T}\bm{1}_{N_\mathrm{A}} \in \mathcal{R}(\bm{H})$.
Under this condition, the coefficients vector that assures~\eqref{eq:linearTiming} is computed as 
\begin{equation}
    \bm{\beta}^\star = - c\gamma_\mathrm{T} \left(\bm{U}_{N_\mathrm{A}}^\mathrm{T} \bm{U}_{N_\mathrm{A}}\right)^{-1} \bm{U}_{N_\mathrm{A}}^\mathrm{T}  \bm{1}_{N_\mathrm{A}}\;.
\end{equation}

Finally, for the generation attack, the attacker has to pick ranges 
\begin{equation}
    \Delta \bm{r}^\star =  c\gamma_\mathrm{T}\bm{1} + \sum_{n = 1}^{N-4}\beta^{\star}_n \bm{u}_n \;.
\end{equation}
Indeed, by construction of $\bm{\beta}^\star$, $\Delta {r}^\star_n = 0$ for $1\leq n \leq N_\mathrm{A}$, hence only the open signals need be tampered. 

Concluding, we showed that if $N_\mathrm{O}\geq 4$, the attacker can exploit the just described procedure to perform a generation-type attack that alters the victim time of a factor $\gamma_\mathrm{T}$, generating only the open signals.

\section{Analysis for the Realistic Scenario}\label{sec:noise}
In this Section, we analyze the performance of the attack considering a realistic scenario that takes into account the noise at the victim and the attacker receivers. 
We remark that, when the attacker performs a generation attack, the fake signals are only affected by the noise added at the victim receiver. 
On the other hand, when the attacker performs a relay attack, the tampered signals are affected by both the noises produced by the victim's and the attacker's receivers.  
We assume that the channels between the satellites and the receivers, of both the attacker and the victim, are \ac{awgn}, as well as the channel between the attacker and the victim. We denote with $\sigma_\mathrm{L}$ and $\sigma_\mathrm{A}$ the noise standard deviations at the victim and attacker receiver, respectively. Moreover, as typically done for \ac{gnss}, we assume each channel to be independent, with equal variance.

The noise affects the pseudorange estimations, thus, the covariance matrices of the legitimate and tampered pseudoranges residuals, respectively, are
\begin{equation}\label{eq:pseud_cov}
    \bm{\Sigma}_\mathrm{L} = \mbox{cov}(\Delta\bm{r}) \triangleq \sigma^2_\mathrm{L} \bm{I}_N\,,\;\;\; 
    \bm{\Sigma}_\mathrm{T} = \mbox{cov}(\Delta\widetilde{\bm{r}}) \triangleq \sigma^2_\mathrm{T} \bm{I}_N\,,
\end{equation}
where $\sigma^2_\mathrm{T} = \sigma^2_\mathrm{L}$ when the attacker performs a generation attack, and $\sigma^2_\mathrm{T} = \sigma^2_\mathrm{L} + \sigma^2_\mathrm{A}$ when he performs a relay attack. Moreover, the legitimate pseudorange residuals are zero-mean, at the convergence of the \ac{pvt} iterative computation. Taking into account the additional tampering introduced by the attacker, the tampered pseudoranges mean is  
\begin{equation}
    \mathbb{E}\left[\Delta \bm{\widetilde{r}} | \mathcal{H}_1\right] =  \mathbb{E}\left[ \Delta \bm{r} + \Delta \bm{r}_\mathrm{T} \right] = \Delta \bm{r}_\mathrm{T}\,.
\end{equation}
In the next Sections, we separately consider the attacks to the timing and the position-based checks.

\subsection{Attacks to the Time-Based Checks}
Our aim is to statistically characterize the vector $\bm{\theta}_\mathrm{t}$, defined in~\eqref{eq:generalCheck}, in realistic conditions.
Indeed, $\bm{\theta}_\mathrm{t} = \bm{C} \bm{H}\Delta \bm{r}$ is a Gaussian vector since it is derived from the linear combination of the Gaussian vector $\Delta \bm{r}$. Thus, to characterize it, we have to compute its mean and covariance, considering both legitimate and under-attack conditions.

In the legitimate scenario, $\mathcal{H}_0$, the mean is
\begin{equation}
    \mathbb{E}[\bm{\theta}_\mathrm{t}| \mathcal{H}_0] =\bm{C H} \mathbb{E}[\Delta \bm{r}| \mathcal{H}_0] = 0\,,
\end{equation}
while the covariance of $\bm{\theta}_\mathrm{t}$ in $\mathcal{H}_0$ is  
\begin{equation}\label{eq:variance_legit}
   \begin{split}
   \bm{\Sigma_{0}} \triangleq \mathbb{E}\left[\bm{\theta}_\mathrm{t} \bm{\theta}_\mathrm{t}^\mathrm{T}| \mathcal{H}_0\right] = \mathbb{E}\left[ \bm{C}\bm{H} \Delta \bm{r}  (\bm{C} \bm{H} \Delta \bm{r} )^\mathrm{T}| \mathcal{H}_0\right] = \\
      \bm{C} \bm{H} \mathbb{E}\left[\Delta \bm{r} \Delta \bm{r} ^\mathrm{T}\right]   \bm{H} ^\mathrm{T}\bm{C} ^\mathrm{T} = \sigma_\mathrm{L}^2 \bm{C}\bm{H}  \bm{H} ^\mathrm{T}\bm{C} ^\mathrm{T}\,,
    \end{split} 
\end{equation}
where, in the last step, we used the definition of $\bm{\Sigma}_L$ given in~\eqref{eq:pseud_cov}.

For the under-attack scenario, $\mathcal{H}_1$, the mean is given by
\begin{equation}\label{eq:meanAttack}
    \mathbb{E}[\bm{\theta}_\mathrm{t} | \mathcal{H}_1] =  \bm{C} \bm{H} \Delta \bm{r}_\mathrm{T} \triangleq \bm{\mu} \,,
\end{equation}
while, analogously to~\eqref{eq:variance_legit}, the covariance of $\bm{\theta}_\mathrm{t}$ in $\mathcal{H}_1$ is 
\begin{equation}
   \begin{split}
    \bm{\Sigma_{1}}\triangleq \mathbb{E}\left[\bm{\theta}_\mathrm{t} \bm{\theta}_\mathrm{t}^\mathrm{T}| \mathcal{H}_1 \right] =  \sigma_\mathrm{T}^2 \bm{C}\bm{H}  \bm{H} ^\mathrm{T}\bm{C} ^\mathrm{T} \,.
    \end{split} 
\end{equation}
where the last equality follows from the definition of $\bm{\Sigma}_T$ given in~\eqref{eq:pseud_cov}. Moreover, the attack to the time-based check described in Section~\ref{sec:genatt_VS_TimeCheck} is a generation attack, thus $\sigma^2_\mathrm{T} = \sigma^2_\mathrm{L}$.

The false alarm and miss detection probabilities are, respectively, 
\begin{align}\label{eq:generalPFA} 
        p_\mathrm{FA} &= \mathrm{P}\hspace*{-0.1cm} \left[ \bigcup_{i = 1}^{N_\mathrm{c}}\{\theta_{\mathrm{t},i} > \delta_{\mathrm{t},i}\} |\mathcal{H}_0\right]\hspace*{-0.1cm} = 1-\mathrm{P}\hspace*{-0.1cm} \left[ \bigcap_{i = 1}^{N_\mathrm{c}}\{\theta_{\mathrm{t},i} \leq \delta_{\mathrm{t},i} \}|\mathcal{H}_0\right]\hspace*{-0.1cm}, \\
    p_\mathrm{MD} &=  \mathrm{P}\hspace*{-0.1cm} \left[ \bigcap_{i = 1}^{N_\mathrm{c}}\{\theta_{\mathrm{t},i} \leq \delta_{\mathrm{t},i}\} |\mathcal{H}_1 \right].
\end{align}
We remark that, in general, metrics $\theta_{\mathrm{t},i}$ and $\theta_{\mathrm{t},j}$, with $i \neq j$, are not statistically independent.

\paragraph*{Attack Against the ISB-based Cross-authentication Check}
We consider now the check of \cite{motellaCrossCheck}, where only Galileo and GPS are considered. First, we neglect the bias $ b + GGTO$, which can be subtracted in advance from one measurement. So, it yields $\delta_{\mathrm{t},1} = \delta_{\mathrm{t},2} = T$. Moreover, from \ref{eq:modelMotella}, $\bm{c}_1 = - \bm{c}_2$, thus also $\theta_{\mathrm{t},1} = - \theta_{\mathrm{t},2}$, and~\eqref{eq:generalPFA} becomes
\begin{equation}\begin{split}
    p_\mathrm{FA}  &=1-P(\theta_{\mathrm{t},1}  < T \wedge \theta_{\mathrm{t},2}  < T| \mathcal{H}_0) \\ &= 1-P( - \delta_1 < \theta_{\mathrm{t},1}  < T  | \mathcal{H}_0)    =2Q(T/\sigma_0)\;,
\end{split}\end{equation}
where $\sigma_0  = \sigma_\mathrm{L}\sqrt{ \bm{c}_1\bm{H}  \bm{H} ^\mathrm{T}\bm{c}_1 ^\mathrm{T}}$. Next, we can invert the above relation to computing the threshold values as a function of the $p_\mathrm{FA}$ set by the receiver,
\begin{equation}
    T = \sigma_0 \mathrm{Q}^{-1}\left( \frac{1}{2} p_\mathrm{FA} \right)\,.
\end{equation}

Analogously, calling $\mu_1$ and $\mu_2$ the first and second element of $\bm \mu$, respectively, the miss detection probability is
\begin{equation}\label{eq:pmd_pfa} \begin{split}
    p_\mathrm{MD} &= P(\theta_{\mathrm{t},1} \leq \delta_{\mathrm{t},1} \wedge \theta_{\mathrm{t},2} \leq \delta_{\mathrm{t},2} | \mathcal{H}_1)  \\
    &= \mathrm{Q}\left( \frac{-T - \mu_1}{\sigma_1} \right) - \mathrm{Q}\left( \frac{T - \mu_1}{\sigma_1} \right)  \\
    &= \mathrm{Q}\left( -\frac{ \sigma_0}{\sigma_1} \mathrm{Q}^{-1}\left( \frac{1}{2} p_\mathrm{FA} \right) - \frac{\mu_1}{\sigma_1} \right) -\\- & \mathrm{Q}\left( \frac{ \sigma_0}{\sigma_1} \mathrm{Q}^{-1}\left( \frac{1}{2} p_\mathrm{FA} \right) - \frac{\mu_1}{\sigma_1} \right) \;,
    \end{split}\end{equation}
where $\mu_1 = - \mu_2 = \bm{c}_1 \bm{H} \Delta \bm{r }_\mathrm{T}$ and $\sigma_1 =  \sigma_\mathrm{T}\sqrt{ \bm{c}_1\bm{H}  \bm{H} ^\mathrm{T}\bm{c}_1 ^\mathrm{T}}$. Thus, \eqref{eq:pmd_pfa} shows the missed detection probability as a function of the false alarm, legitimate and attacker channel noises, and $\bm\mu$. In particular, $\bm\mu$ is also related to the range alteration induced by the attacker, which is a function of the distance between the true solution $\bm{p}$ and the target $\bm{p}^\star$.

In a generation attack scenario, where $\sigma_0 = \sigma_1$, for small values of $\mu_1$, from \ref{eq:pmd_pfa} we get $p_\mathrm{MD} = p_\mathrm{FA} $, which means that fake signals are indistinguishable from the legitimate ones. This shows that the success probability depends on $\mu_1$, the amount of displacement induced by the attacker, so that the more the attacker tries to tamper with the signals, the easier it is for the victim to detect the attack.
For instance, if we consider a relay attack for position spoofing, as \textit{selective delay}, then $\sigma_1 \gg \sigma_0$, i.e., the noise standard deviation of the tampered pseudoranges is higher than that of the legitimate ones, $p_\mathrm{MD} \rightarrow 0$ for any chosen $p_\mathrm{FA}$, thus the receiver is always able to detect the attack.
In Section~\ref{sec:ResNoise} we will provide a numerical evaluation of these observations.

\subsection{Attacks to the Position-Based Checks}\label{sec:posCheck_noise}
Our aim is to statistically characterize the term ${\theta}_\mathrm{pos}$ in~\eqref{eq:testTiming}, in a realistic environment, i.e., when the measurements are affected by noise. In particular, we aim at relating security performance, i.e., false and miss detection, to the channel noise in the legitimate and under-attack scenarios, $\sigma_\mathrm{L}^2$ and $\sigma_\mathrm{T}^2$. Notice that, in the legitimate case, this problem is equivalent to the general problem of characterizing the position accuracy for \ac{gnss}.

From~\eqref{eq:testTiming}, we can write
\begin{equation}\label{eq:thetasquare}
    \theta_\mathrm{pos}^2 = \bm{\epsilon}^\mathrm{T}\bm{\epsilon}\;,
\end{equation}
where $\bm{\epsilon}$ is the 3D position displacement with respect to the reference position. We remark that the reference position with coordinates ($x_\mathrm{ref},y_\mathrm{ref},z_\mathrm{ref}$) is assumed to be known without error.
Thus, to characterize the performance of the test~\eqref{eq:testTiming} we model term $\bm{\epsilon}$ first. Indeed, $\bm{\epsilon}$ is a Gaussian vector since the position is derived as a linear combination of the Gaussian random vector $ \Delta \bm{r}$, as in~\eqref{eqn:pHr}. 
In the legitimate case, we have 
\begin{equation}\label{eq:meanXi}
    \mathbb{E}[\bm{\epsilon}| \mathcal{H}_0] = 0\;,
\end{equation}
since the \ac{pvt} solution is expected to converge to the reference position itself. For the covariance we have
\begin{equation}\label{eq:covXi}
    \mathbb{E}[\bm{\epsilon} \bm{\epsilon} ^\mathrm{T} | \mathcal{H}_0] = \sigma^2_\mathrm{L} \bm{H}_{1-3} \bm{H}_{1-3}^\mathrm{T}\;,
\end{equation}
where the subscript of $\bm{H}_{1-3}$ refers to the fact that we are considering only rows 1, 2, and 3 of $\bm{H}$.

Under attack, we have
\begin{equation}
    \mathbb{E}[\bm{\epsilon}| \mathcal{H}_1] = \bm{H}_{1-3}  \Delta \bm{r}_\mathrm{T} \;,
\end{equation}
and, for the covariance,
\begin{equation}
     \mathbb{E}[\bm{\epsilon} \bm{\epsilon} ^\mathrm{T} | \mathcal{H}_0] = \sigma^2_\mathrm{T} \bm{H}_{1-3} \bm{H}_{1-3}^\mathrm{T}\;.
\end{equation}
Note that \eqref{eq:thetasquare} is a quadratic form, so we can take advantage of the following Theorem.
\begin{theorem}
    The random variable $\theta_\mathrm{pos}^2 = \bm{\epsilon}^\mathrm{T}\bm{\epsilon}$, with $\bm{\epsilon} \sim \mathcal{N}\left(\bm{\mu}, \bm{\Sigma} \right)$, can be written as 
    \begin{equation}
        \theta_\mathrm{pos}^2 = \left(\bm{u} + \bm{b} \right) \bm{\Lambda} \left(\bm{u} + \bm{b}\right) = \sum_i^3 \lambda_i (u_i + b_i)^2\;,
    \end{equation}
    where $\bm{u}\sim \mathcal{N}\left(\bm{0}, \bm{I}_3\right)$, $\bm{\Lambda}$ is the diagonal matrix obtained from the decomposition $ \bm{\Sigma} = \bm{P} \bm{\Lambda} \bm{P} ^\mathrm{T}$ and $\bm{b}^\mathrm{T} = \bm{P}^\mathrm{T}  \bm{\Sigma} ^{-\frac12} \bm{\mu}$. Hence, $\theta_\mathrm{pos}^2$ can be expressed as the linear combination of non-central chi-square random variables of degree 1, with non-centrality parameter $b_i^2$.    
    \end{theorem}
\begin{proof}
    The proof can be found in~\cite[Ch. 3]{mathai92}, for the quadratic form $Q(\bm{X}) = \bm{X}^\mathrm{T} \bm{A} \bm{X}$, but using $\bm{A} = \bm{I}$.
\end{proof}

Indeed, in the legitimate case, $\epsilon$ has zero mean, thus, $\theta^2$ is actually the linear combination of (central) chi-square random variables of degree $1$.
Still, it is not possible to derive the closed-form expression for either the false alarm or miss detection probabilities, thus relating the metric value, $\theta_\mathrm{pos}$ to the chosen threshold. 

For an application-oriented (but approximate) derivation of such relation we could exploit the definition of \ac{gdop} \cite[Ch. 7]{kaplan}. First, we rotate the reference system to have the first three coordinates pointing to the \ac{enu} reference frame. Next, repeating the steps in~\eqref{eq:covXi}, we get as position error
\begin{equation}
    \mathbb{E}\left[\Delta \bm{p} \Delta \bm{p} ^\mathrm{T} \right]= \sigma^2_L \bm{H} \bm{H}^\mathrm{T} = \sigma^2_L \left(\bm{G}^\mathrm{T} \bm{G}\right)^{-1}.
\end{equation}
Hence, taking the trace operation, we define the \ac{gdop} as 
\begin{equation}
    \mathrm{GDOP} \triangleq  \frac{1}{\sigma_L}\sqrt{\sigma_\mathrm{E}^2 + \sigma_\mathrm{N}^2 +\sigma_\mathrm{U}^2  + \sigma_\mathrm{t}^2 }\;,
\end{equation}
where $\sigma_\mathrm{E}^2$, $\sigma_\mathrm{N}^2$, and $\sigma_\mathrm{U}^2$ are the variances associated respectively with the error along the East, North, and vertical directions while $\sigma_\mathrm{t}^2$ is the error for the clock bias (measured in meters). To focus on the position error, it is useful to consider the \ac{pdop} 
\begin{equation}
    \mathrm{PDOP} \triangleq  \frac{1}{\sigma_L}\sqrt{\sigma_\mathrm{E}^2 + \sigma_\mathrm{N}^2 +\sigma_\mathrm{U}^2  }\;,
\end{equation}
These terms are used to model, and therefore also to predict, the positioning error since the \ac{dop} values are related to the pseudorange noise and the geometry. Thus, low \ac{dop} values are typically associated with high \ac{pvt} accuracy. An additional in-depth discussion of the \ac{dop} and their derivation can be found in~\cite[Ch. 21]{teunissen2017}.

Hence, it is possible to use this characterization to describe $p(\theta_\mathrm{pos}|\mathcal{H}_0)$ and $p(\theta_\mathrm{pos}|\mathcal{H}_1)$. Still, these estimates would be based only on an approximate model which is not suited for security analysis. The performance of the position check will be therefore evaluated only numerically, in Section~\ref{sec:ResNoise}.

\section{Numerical Results}\label{sec:results}
In this Section, we assess the performance of the proposed attack against the authentication mechanism discussed in Section~\ref{sec:verModel}. 
More in detail, first we describe the data collection procedure and the validation phase. Next, we test both the attacks of Section~\ref{sec:attack}, considering both noiseless and noisy scenarios.

\subsection{Data Collection}
The dataset was built by collecting signals from a quasi-urban environment with a Septentrio PolarRx5 receiver connected to an A42 Hemisphere antenna. In particular, the measurements were collected from our lab window sill ($45.408^\circ\,$N, $11.894^\circ\,$E, altitude \SI{30}{\meter} above the sea level).

The receiver output was post-processed obtaining a dataset of measurements from different constellations. We focused on GPS L1 C/A and Galileo E1 BC. 
The dataset, contained \SI{10}{\minute} of measurements, collected with a frequency of \SI{1}{\hertz}, for a total of 600 \ac{pvt} epochs.
We consider the GPS signals as authenticated (e.g., by CHIMERA) and the Galileo signals as open, and we collected signals from $N = 8$ different satellites in view, with $N_\mathrm{A} = 3$ and $N_\mathrm{O} = 5$.

More in detail, for each considered epoch, the dataset contained: the pseudoranges measurements, the satellite clock biases, the tropospheric and the ionospheric delays (estimated using the Klobuchar model \cite{klobuchar87}), and the satellite positions. Satellite clock biases, atmospheric delays, and satellite positions were derived from the respective navigation message. 

To test the performance of the attacks, we implemented a \ac{pvt} computation block operating as summarized in the Appendices, following the description of \cite[Ch. 8]{borre07}.

\begin{figure}
    \centering    
    \includegraphics[width=\columnwidth]{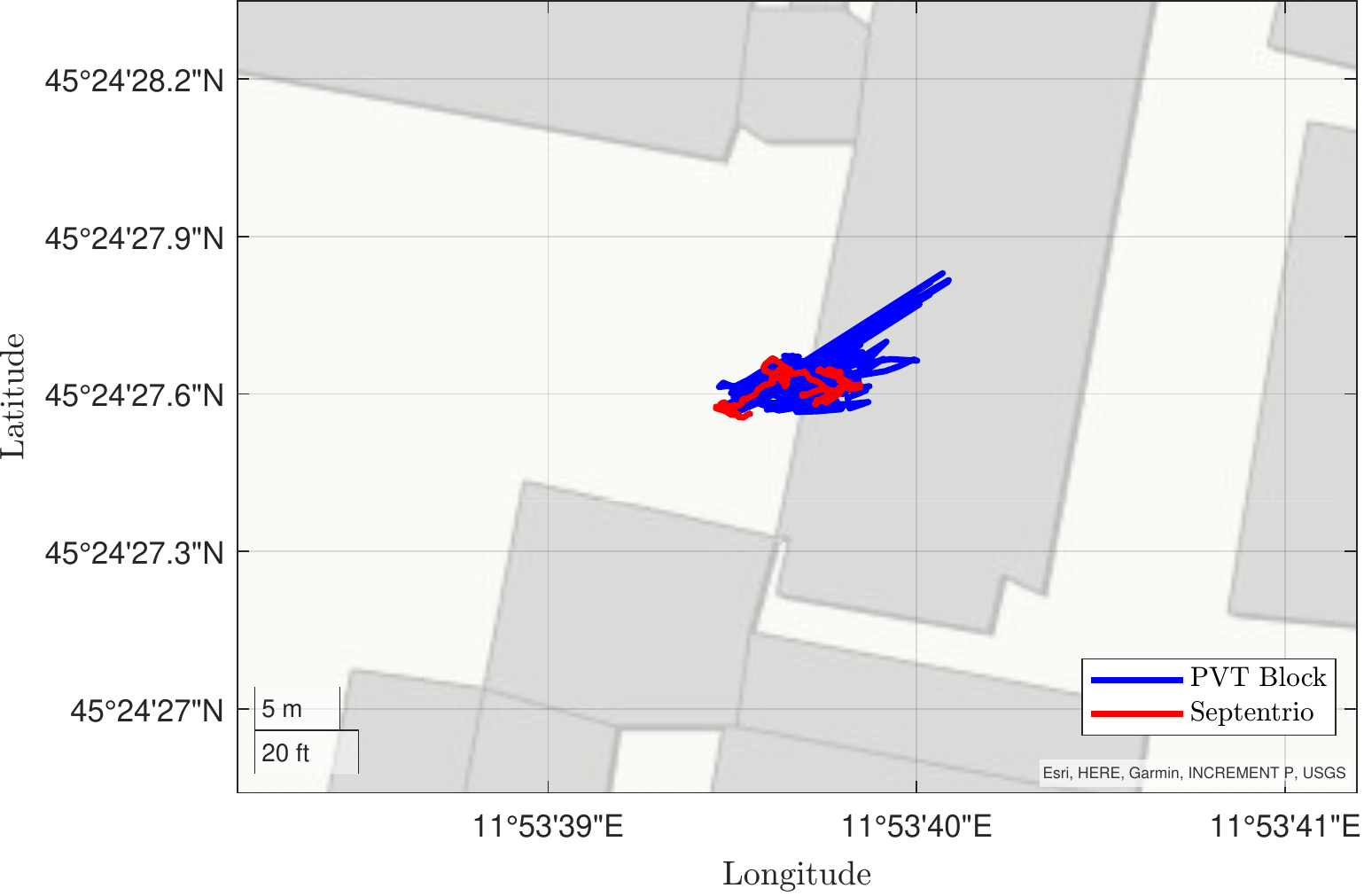}
    \caption{Validation of our \ac{pvt} block: position computed using our \ac{pvt} block (blue) and the Septentrio's software (red).}
    \label{fig:focus}
\end{figure}

To validate our \ac{pvt} block, we report in Fig.~\ref{fig:focus} respectively the positions computed by the Septentrio receiver, in blue, and those computed by our \ac{pvt} block, in red. The positions obtained from the Septentrio receiver are slightly more precise than ours: the error standard deviations were \SI{3.03}{\meter} and \SI{3.13}{\meter} using, respectively, the Septentrio receiver and our \ac{pvt} block.

\subsection{Attacks in Ideal Scenario Conditions}
We now present the attack performance in the noiseless scenario, thus $\sigma_\mathrm{L} = \sigma_\mathrm{T} =0$, distinguishing between the attack to the time-based and the position-based checks.

\paragraph{Attack Against the Time-Based Check}
We present the results of the proposed attack. In particular, we use the following procedure, where at each epoch we
\begin{enumerate}
    \item compute the \ac{pvt} solution by using the legitimate measurements (legitimate case);
    \item chose the target position $\bm{p}^\star_1$;
    \item solve  problem~\eqref{eq:minProb}, obtaining the tampered ranges $\Delta \bm{\widetilde{r}}$;
    \item feed the tampered ranges to the \ac{pvt} block;
    \item check the authenticity of the solution using the \textit{\ac{isb}-based cross-authentication check}~\eqref{eqn:test}.
\end{enumerate}
In particular, we set the target position $\bm{p}^\star_1$ at \SI{45.398}{\degree} N, \SI{11.876}{\degree} E, with altitude \SI{12}{\meter} (Prato della Valle square, Padua, Italy).
The distance between the legitimate and the target position is roughly \SI{1.7}{\kilo\meter}. 

Notice that, between steps 3) and 4) we omitted the actual fake signal generation and the processing of the victim on the received signal. This means that we assumed that, under attack, the victim processes the tampered signal as intended by the attacker. Still as experimentally proved, e.g., in~\cite{ceccato18} and \cite{lenhart21}, it is indeed feasible to induce to the victim the fake signals, hence we skipped these steps.

\begin{figure}
    \centering   
    \includegraphics[width=1.01\columnwidth]{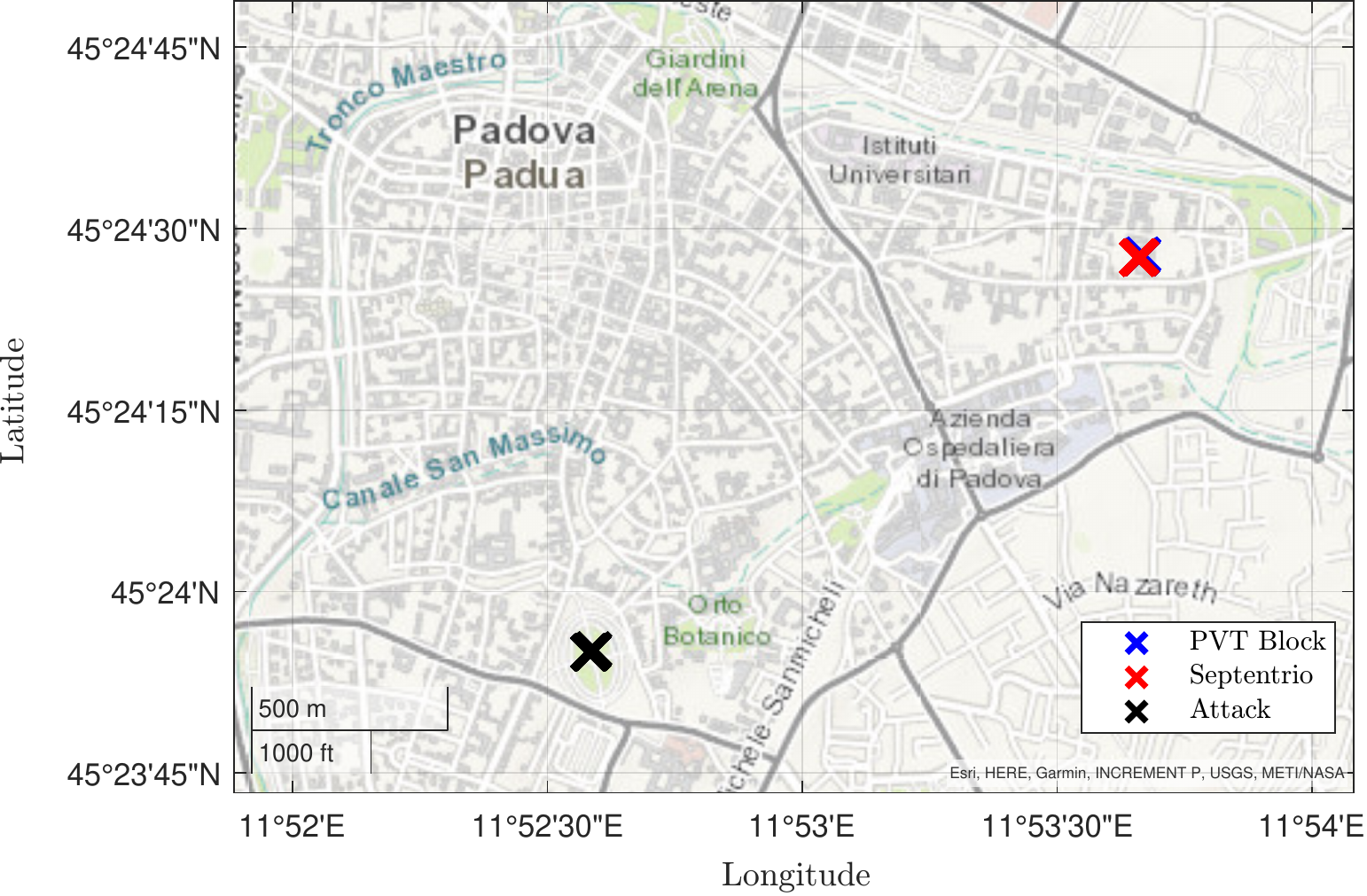}
    \caption{Positions computed in the legitimate (blue) and under attack scenario (black), compared to the ground truth (red). The position computed in the under-attack scenario coincides with the target $\bm{p}^\star$.}
    \label{fig:attackMap}
\end{figure}

Fig.~\ref{fig:attackMap} shows the positions computed in the legitimate and under-attack scenarios: the attacker is always able to induce the target position and only negligible errors are observed with respect to the target position. No issue has been observed in terms of availability, i.e., all the fake ranges led to the computation of a \ac{pvt}.

\begin{figure}
    \centering   
    \input{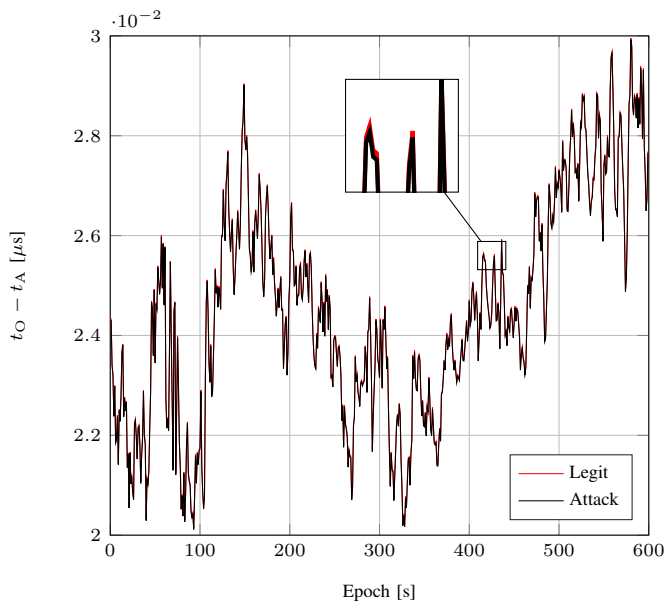}
    \caption{Metric for the check~\eqref{eqn:test}: the metric computed under attack (black) is superimposed to the one computed in the legitimate case (red).}
\label{fig:motellaCheckUnderAttack}
\end{figure}

Fig.~\ref{fig:motellaCheckUnderAttack} shows the difference between the clock biases $t_\mathrm{O}$ and $t_\mathrm{A}$ used in the security check~\eqref{eqn:test}, computed respectively in the legitimate and under-attack case. Indeed, the check cannot distinguish the legitimate from the under-attack scenario, since the actual differences between legitimate and fake are well below the standard deviation of the metric itself. Hence, for any threshold $T$, the check cannot effectively detect the attack.
For these reasons, we conclude that our attack can spoof the victim's position without being detected. 

Fig.~\ref{fig:roc_distance} shows instead the \ac{det} curves for different distances between target and actual position. Indeed as this distance increases,  it becomes easier for the victim to detect the spoofing attack. In particular, the attack success probability rapidly decreases for distances higher than \SI{10}{\kilo\meter}. This is in line with our observations in Section~\ref{sec:attackerModel}, confirming that the linearization procedure is not reliable to design the attack if the target position is too far from the real one. Still, we remark that it may not be reasonable for the attacker to consider the target position too far since other signals of opportunity may be exploited by the victim to check the computed position.

\begin{figure}
    \centering
    \vspace{0.26cm}
    \input{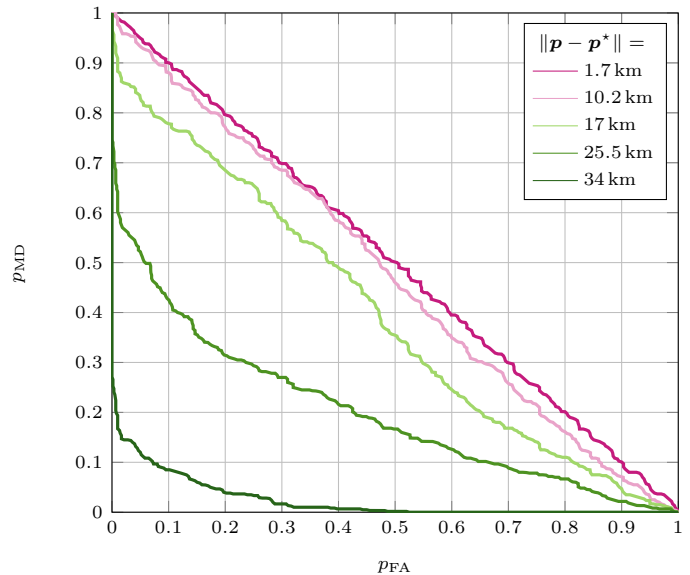}
    \caption{\ac{det} curves for different distances between legit and target positions, for the time-based cross-authentication against the proposed generation attack.}
    \label{fig:roc_distance}
\end{figure}

\paragraph{Attack Against the Position-Based Check}
We report the results of the attacks against the position cross-authentication checks, discussed in Section~\ref{sec:attTimePositionCheck} for the noiseless scenario.
In particular, three different attacks are considered (Table \ref{tab:time_push}), with different induced time shifts and attack start times.
\begin{table}
    \centering 
    \caption{Parameters of the attack to the position cross authentication check: attack starts $t_\mathrm{start}$ and additional clock biases $\Delta t -\Delta t_\mathrm{L}$. }
    \begin{tabular}{c>{\columncolor[HTML]{EFEFEF}}cc>{\columncolor[HTML]{EFEFEF}}cc>{\columncolor[HTML]{EFEFEF}}c}\toprule
         & Attack 1 & Attack 2 & Attack 3  \\ \midrule
        $t_\mathrm{start}$  & \SI{60}{\second} & \SI{100}{\second} & \SI{150}{\second}\\ 
         $\Delta t -\Delta t_\mathrm{L}$ & \SI{30}{\mu \second} & \SI{10}{\mu \second} & \SI{5}{\mu \second} \\ \bottomrule
    \end{tabular}
    \label{tab:time_push}
\end{table}

Fig.~\ref{fig:time_push_ide} shows the results of the step-change on the clock bias estimated by the victim. This represents a best-case scenario for the defender, since, typically, the time shift is induced in a ramp-like manner (e.g., see the classification of \cite{Schmidt21}). Indeed, if the considered attacks are successful, also the more sophisticated attacks will be successful as well.
During the tests, all the considered attacks achieve their aim, inducing the target time shifts. On the other hand, consider the actual position check, no significant position deviation is induced by any of the attacks, since the deviation is at the cm-level, well below the actual accuracy of the receiver. Again, none of these attacks would be detected by the victim using the considered cross-authentication checks.

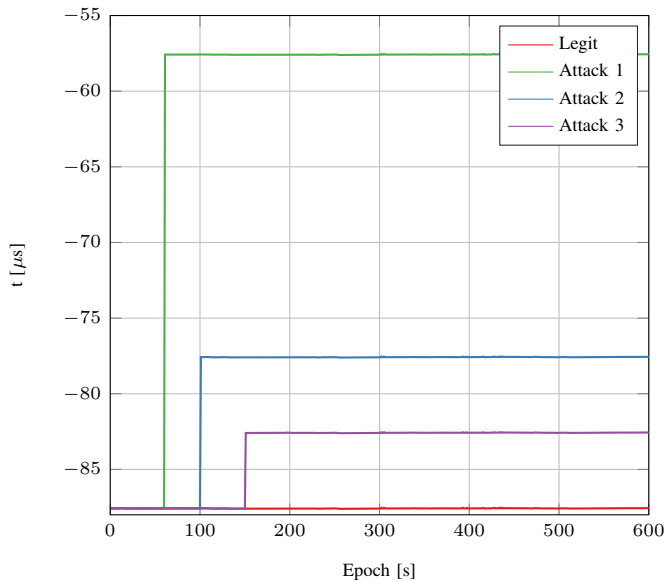
\begin{figure}
    \centering
    \definecolor{mycolor1}{RGB}{228,26,28}%
\definecolor{mycolor2}{RGB}{55,126,184}%
\definecolor{mycolor3}{RGB}{77,175,74}%
\definecolor{mycolor4}{RGB}{152,78,163
}%

\begin{tikzpicture}[every plot/.style={thick}]

\begin{axis}[%
width=0.951\fwidth,
height=\fheight,
at={(0\fwidth,0\fheight)},
scale only axis,
xmin=0,
xmax=600,
xlabel style={font=\color{white!15!black}},
xlabel={Epoch [s]},
ymin=-88,
ymax=-55,
ytick distance = 5,
ylabel style={font=\color{white!15!black}},
ylabel={t [$\mu$s]},
axis background/.style={fill=white},
xmajorgrids,
ymajorgrids,
legend style={legend cell align=left, align=left, draw=white!15!black},
enlargelimits=false,title style={font=\scriptsize},xlabel style={font=\scriptsize},ylabel style={font=\scriptsize},legend style={font=\scriptsize},ticklabel style={font=\scriptsize}
]
\addplot [color=mycolor1]
  table[row sep=crcr]{%
1	-87.5739832802125\\
5	-87.5747345898327\\
7	-87.5780190615529\\
11	-87.5638273977804\\
19	-87.5753221783334\\
104	-87.5732704605904\\
107	-87.5690528008622\\
114	-87.5901545368663\\
121	-87.5875117340162\\
126	-87.5815352555452\\
129	-87.5847683707633\\
136	-87.5950436472762\\
229	-87.5889237002749\\
233	-87.5980613276456\\
238	-87.5869191966424\\
241	-87.5896303614703\\
245	-87.5844838613302\\
251	-87.5793191158826\\
254	-87.5984466208179\\
257	-87.6065497366949\\
301	-87.5882010205412\\
302	-87.5429409547389\\
303	-87.5841558092775\\
304	-87.5850850418839\\
305	-87.5401501698876\\
306	-87.5787688075284\\
310	-87.5819756367382\\
329	-87.5779543182675\\
331	-87.5792359566123\\
338	-87.5716333982353\\
347	-87.5808552981415\\
351	-87.5759806636332\\
356	-87.5818774196641\\
365	-87.5756032595265\\
380	-87.5751999226873\\
383	-87.5789154210112\\
392	-87.5755594193296\\
393	-87.5443840121966\\
394	-87.5741516420539\\
396	-87.577850268748\\
397	-87.5490641148311\\
398	-87.5845308101323\\
400	-87.5774433882641\\
415	-87.5722688513575\\
416	-87.5455747705387\\
417	-87.5834484438277\\
420	-87.5750544502013\\
425	-87.5761369077386\\
426	-87.5638097502476\\
427	-87.5655005433183\\
428	-87.5392070189654\\
429	-87.5711271606353\\
433	-87.565377563212\\
434	-87.5349364892312\\
435	-87.5720954232909\\
436	-87.5682047701649\\
437	-87.5402912167983\\
438	-87.5634608075766\\
463	-87.5701050807869\\
470	-87.5777832582216\\
473	-87.5633194073489\\
476	-87.5663885697144\\
480	-87.5760375848889\\
488	-87.5783661403765\\
490	-87.5816197626797\\
493	-87.5779961111207\\
497	-87.5867914677669\\
506	-87.5728016107216\\
508	-87.5741711844805\\
511	-87.5850684864101\\
518	-87.5856761962713\\
557	-87.5719946449216\\
567	-87.5667856711402\\
583	-87.5649328909454\\
589	-87.5650944052893\\
600	-87.5587130998262\\
};
\addlegendentry{Legit}
\addplot [color=mycolor3]
  table[row sep=crcr]{%
1	-87.5739832802125\\
5	-87.5747345898327\\
7	-87.5780190615529\\
11	-87.5638273977804\\
19	-87.5753221783334\\
60	-87.5813376182728\\
61	-57.5793004128393\\
63	-57.580504743028\\
87	-57.5803868065605\\
103	-57.574978742357\\
109	-57.5779995512906\\
121	-57.5875117340261\\
126	-57.581535255572\\
129	-57.5847683707808\\
136	-57.5950436472768\\
229	-57.5889237002892\\
233	-57.5980613276452\\
238	-57.5869191966221\\
241	-57.5896303614828\\
245	-57.5844838612858\\
251	-57.5793191159047\\
254	-57.5984466207912\\
257	-57.6065497367146\\
301	-57.5882010205327\\
302	-57.5429409546692\\
303	-57.5841558093066\\
304	-57.5850850418865\\
305	-57.5401501698145\\
306	-57.5787688075619\\
310	-57.5819756367495\\
329	-57.5779543182414\\
331	-57.5792359566052\\
338	-57.5716333982788\\
347	-57.5808552981553\\
351	-57.5759806636488\\
356	-57.5818774196772\\
365	-57.575603259522\\
380	-57.5751999227188\\
383	-57.578915420993\\
392	-57.575559419318\\
393	-57.5443840121875\\
394	-57.5741516420377\\
396	-57.577850268761\\
397	-57.5490641148504\\
398	-57.5845308101132\\
400	-57.5774433882741\\
415	-57.5722688513697\\
416	-57.5455747705246\\
417	-57.5834484438154\\
420	-57.5750544502346\\
425	-57.5761369077391\\
426	-57.5638097502678\\
427	-57.5655005433023\\
428	-57.5392070189478\\
429	-57.5711271606722\\
433	-57.5653775631774\\
434	-57.5349364892226\\
435	-57.5720954232571\\
436	-57.5682047701616\\
437	-57.5402912168496\\
438	-57.5634608075633\\
463	-57.5701050807958\\
470	-57.5777832582163\\
473	-57.5633194073782\\
476	-57.5663885697212\\
480	-57.5760375849238\\
488	-57.57836614037\\
490	-57.5816197626619\\
493	-57.5779961111175\\
497	-57.5867914677693\\
506	-57.5728016107034\\
508	-57.5741711844997\\
511	-57.5850684863941\\
518	-57.5856761962521\\
557	-57.5719946449143\\
567	-57.5667856711401\\
583	-57.5649328909064\\
589	-57.565094405262\\
600	-57.5587130998073\\
};
\addlegendentry{Attack 1}
\addplot [color=mycolor2]
  table[row sep=crcr]{%
1	-87.5739832802125\\
5	-87.5747345898327\\
7	-87.5780190615529\\
11	-87.5638273977804\\
19	-87.5753221783334\\
100	-87.5792327894726\\
101	-77.5801027623899\\
105	-77.5696473896272\\
108	-77.5736996183857\\
114	-77.5901545368678\\
121	-77.5875117340282\\
126	-77.5815352555743\\
129	-77.5847683707831\\
136	-77.5950436472791\\
229	-77.5889237002914\\
233	-77.5980613276473\\
238	-77.5869191966244\\
241	-77.589630361485\\
245	-77.584483861288\\
251	-77.5793191159069\\
254	-77.5984466207934\\
257	-77.6065497367169\\
301	-77.5882010205351\\
302	-77.5429409546716\\
303	-77.584155809309\\
304	-77.5850850418889\\
305	-77.5401501698168\\
306	-77.5787688075643\\
310	-77.5819756367517\\
329	-77.5779543182437\\
331	-77.5792359566075\\
338	-77.5716333982811\\
347	-77.5808552981575\\
351	-77.5759806636511\\
356	-77.5818774196795\\
365	-77.5756032595243\\
380	-77.5751999227209\\
383	-77.5789154209953\\
392	-77.5755594193203\\
393	-77.5443840121898\\
394	-77.57415164204\\
396	-77.5778502687632\\
397	-77.5490641148526\\
398	-77.5845308101154\\
400	-77.5774433882764\\
415	-77.572268851372\\
416	-77.5455747705269\\
417	-77.5834484438178\\
420	-77.5750544502367\\
425	-77.5761369077413\\
426	-77.5638097502699\\
427	-77.5655005433046\\
428	-77.5392070189502\\
429	-77.5711271606745\\
433	-77.5653775631796\\
434	-77.5349364892248\\
435	-77.5720954232594\\
436	-77.5682047701638\\
437	-77.5402912168519\\
438	-77.5634608075654\\
463	-77.570105080798\\
470	-77.5777832582186\\
473	-77.5633194073805\\
476	-77.5663885697235\\
480	-77.5760375849261\\
488	-77.5783661403722\\
490	-77.5816197626642\\
493	-77.5779961111198\\
497	-77.5867914677716\\
506	-77.5728016107057\\
508	-77.574171184502\\
511	-77.5850684863964\\
518	-77.5856761962543\\
557	-77.5719946449166\\
567	-77.5667856711424\\
583	-77.5649328909087\\
589	-77.5650944052643\\
600	-77.5587130998094\\
};
\addlegendentry{Attack 2}

\addplot [color=mycolor4]
  table[row sep=crcr]{%
1	-87.5739832802125\\
5	-87.5747345898327\\
7	-87.5780190615529\\
11	-87.5638273977804\\
19	-87.5753221783334\\
104	-87.5732704605904\\
107	-87.5690528008622\\
114	-87.5901545368663\\
121	-87.5875117340162\\
126	-87.5815352555452\\
129	-87.5847683707633\\
136	-87.5950436472762\\
150	-87.5936000155218\\
151	-82.5928019336955\\
163	-82.5909900757788\\
192	-82.5819202002448\\
195	-82.5870662628148\\
196	-82.5716382580719\\
198	-82.5726764484079\\
204	-82.5897637154159\\
207	-82.5871074989516\\
211	-82.5871162943591\\
215	-82.5851808279994\\
219	-82.5851343363593\\
229	-82.5889237002921\\
233	-82.598061327648\\
238	-82.5869191966251\\
241	-82.5896303614857\\
245	-82.5844838612887\\
251	-82.5793191159075\\
254	-82.598446620794\\
257	-82.6065497367175\\
301	-82.5882010205355\\
302	-82.5429409546721\\
303	-82.5841558093096\\
304	-82.5850850418894\\
305	-82.5401501698175\\
306	-82.5787688075649\\
310	-82.5819756367523\\
329	-82.5779543182442\\
331	-82.579235956608\\
338	-82.5716333982815\\
347	-82.5808552981581\\
351	-82.5759806636516\\
356	-82.58187741968\\
365	-82.5756032595249\\
380	-82.5751999227215\\
383	-82.5789154209958\\
392	-82.5755594193209\\
393	-82.5443840121903\\
394	-82.5741516420406\\
396	-82.5778502687638\\
397	-82.5490641148533\\
398	-82.584530810116\\
400	-82.5774433882771\\
415	-82.5722688513725\\
416	-82.5455747705274\\
417	-82.5834484438184\\
420	-82.5750544502373\\
425	-82.5761369077419\\
426	-82.5638097502706\\
427	-82.565500543305\\
428	-82.5392070189507\\
429	-82.571127160675\\
433	-82.5653775631802\\
434	-82.5349364892254\\
435	-82.57209542326\\
436	-82.5682047701644\\
437	-82.5402912168524\\
438	-82.563460807566\\
463	-82.5701050807987\\
470	-82.5777832582191\\
473	-82.5633194073811\\
476	-82.5663885697242\\
480	-82.5760375849266\\
488	-82.5783661403729\\
490	-82.5816197626648\\
493	-82.5779961111202\\
497	-82.5867914677723\\
506	-82.5728016107064\\
508	-82.5741711845026\\
511	-82.5850684863971\\
518	-82.5856761962548\\
557	-82.5719946449171\\
567	-82.566785671143\\
583	-82.5649328909093\\
589	-82.5650944052649\\
600	-82.5587130998101\\
};
\addlegendentry{Attack 3}

\end{axis}
\end{tikzpicture}%
    \vspace*{-0.8cm}
    \caption{Effect of the time attack on the clock bias estimated by the victim: $\Delta t$ for the legitimate case (red), for Attack 1 (green), 2 (blue), and 3 (purple), whose parameters are reported in Table \ref{tab:time_push}. }
    \label{fig:time_push_ide}
\end{figure}

\subsection{Results with Noisy Measurements}\label{sec:ResNoise}
In this Section, we evaluate the performance of the proposed attack in noisy conditions, distinguishing between time-based and position-based checks. 


Starting from the experimental dataset, we run a Montecarlo simulation, repeating each experiment 35 times, therefore obtaining in total 21\,000 \ac{pvt} measurements.


\paragraph{Attacks Against the Time-Based Check}
As proved in Section~\ref{sec:genatt_VS_TimeCheck}, we show that it is possible for the attacker to lead a successful attack using a generation attack. 

We repeat the same attack simulation procedure adopted for the noiseless scenario but, between steps 3) and 4), we add \ac{awgn} to both the legitimate and the fake pseudoranges, according to the noise standard deviation at the victim receiver $\sigma_\mathrm{L}$, with $\sigma_\mathrm{L}= 0,1,2,4,$ and \SI{9}{\meter}. We remark that, in this case, the attacker performs a generation attack, thus pseudoranges are only affected by the noise at the victim receiver. In general, a reasonable value for the noise standard deviation is $\SI{7.1}{\meter}$ \cite[Ch. 7]{kaplan}\footnote{For instance, from our measurements, we get slightly more than \SI{3}{\meter}.}.
To evaluate the attack performance, we consider a target position \SI{25.5}{\kilo\meter} away from the victim receiver location.

Fig.~\ref{fig:roc} reports the \ac{det} curve for the generation attack against time-based checks, showing the miss detection probability $p_\mathrm{MD}$ as a function of the false alarm probability $p_\mathrm{FA}$ and various values of $\sigma_\mathrm{L}$. 
We notice that, as the noise at the victim receiver increases, the victim becomes unable to distinguish between legitimate and fake signals. Still, the results confirm that the check is not fully reliable even when the noise is minimum since, even in this condition, to get $p_\mathrm{MD}<10^{-2}$ it would yield $p_{\mathrm{FA}} \approx 0.88$, thus rejecting most of the legitimate signals.

\begin{figure}
    \centering
    \input{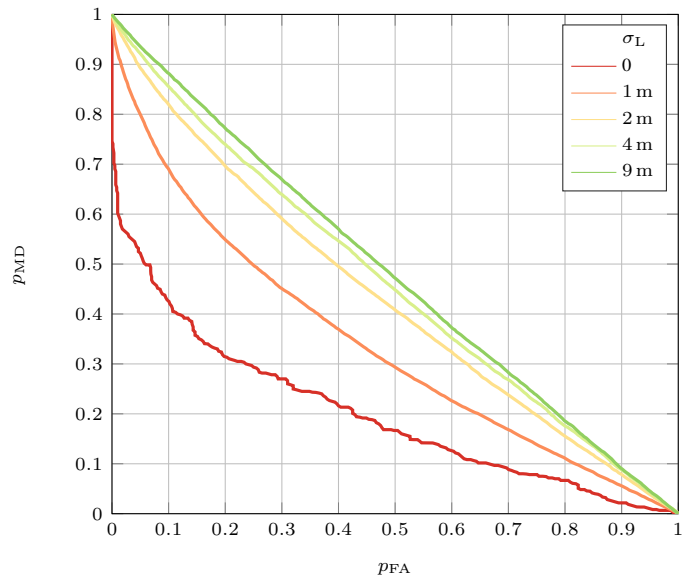}
    \caption{\ac{det} curves with for the generation attack against time-based checks, when the attack target position is \SI{25.5}{\kilo\meter} away from the victim location, for different values of  $\sigma_\mathrm{L}$. }\label{fig:roc}
\end{figure}

\paragraph{Attacks Against the Position-Based Check}
As discussed in Sections~\ref{sec:attTimePositionCheck} and \ref{sec:generationAttack}, if $N_\mathrm{O}<4$ only a relay attack can induce the desired clock-shift while keeping low the probability of being detected. Conversely, when $N_\mathrm{O}\geq 4$, also a generation attack can be used to achieve the same goal. It is worth noting that, when both the relay and the generation attack are viable, the latter is always the preferred option, since it involves signal generation instead of retransmission, and, as a result, the attacker does not introduce additional noise to achieve a better performance in terms of $p_\mathrm{FA}$ and $p_\mathrm{MD}$, as will be discussed in this Section.

First, we consider the relay-type attack described in Section~\ref{sec:attTimePositionCheck}. Indeed, this represents the worst case for the attacker. To show the attack performance, we add noise only to the tampered ranges, so we consider $\sigma_\mathrm{L}=0$. Thus, the fake pseudoranges variance is $\sigma^2_\mathrm{T}= \sigma^2_\mathrm{A}$.

We consider the scenarios where the attacker aims to induce a shift in the estimated clock bias $\Delta t -\Delta t_\mathrm{L}= \SI{10}{\micro \second}$. No relevant difference was observed when considering $\Delta t -\Delta t_\mathrm{L} = 5$ and $\SI{30}{\micro\second}$. In all the considered scenarios, when receiving the tampered signal, the victim computes the clock bias set by the attacker, with only negligible error. Hence whenever no alarm was raised, the attacks were successful. 

Fig.~\ref{fig:time_push_noise} reports the \ac{det} curves of the attack for several values of $\sigma_\mathrm{A}$. Indeed, for $\sigma_\mathrm{A} =0$, thus simulating an attacker able to perform a relay attack without adding any additional noise, tampered and legitimate signals are indistinguishable. On the other hand, as $\sigma_\mathrm{A}$ increases, it is easier for the victim to detect the attack. 
Still, even when $\sigma_\mathrm{A} = \SI{9}{\meter}$, the check is only partially effective since to get low values of missed detection the victim experience high false alarm: for instance to get $p_\mathrm{MD}<10^{-3}$ (not such a low value for a security application) the victim would obtain a $p_{\rm FA} > 0.2$.

\begin{figure}
    \centering
    \input{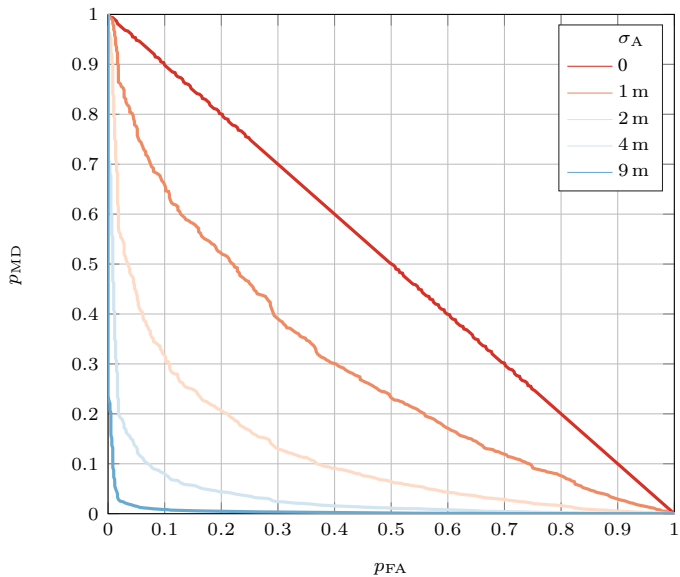}
    \caption{\Ac{det} curves for the relay attack against the position cross-authentication check, for $\Delta t -\Delta t_\mathrm{L}= \SI{10}{\micro \second}$, for several values $\sigma_\mathrm{A}$ and $\sigma_\mathrm{L}$ fixed.} 
    \label{fig:time_push_noise}
\end{figure}

Concerning the generation attack against the time-based check, we carried out attack simulations adding the \ac{awgn} of the victim receiver to both fake and legitimate ranges, with $\sigma_\mathrm{L} = 0,1,2,4,$ and \SI{9}{\meter}. For all the considered values of noise standard deviation, the fake ranges resulted to be indistinguishable from the legitimate ones, and all the \ac{det} curves collapsed to the line $p_{\rm FA} = p_{\rm MD}$.  

\section{Conclusion}\label{sec:conclusions}
In this work, general attack strategies targeting cross-authentication \ac{pvt}-based mechanisms were presented. In such mechanisms, the victim receiver computes the \ac{pvt} using all the available signals, both open and authenticated and tries to assess the authenticity of the computed \ac{pvt} using some a priori information.
In particular, we focused on two classes of checks: the timing cross-authentication, such as \cite{motellaCrossCheck}, where the check is performed verifying if the clock biases relative to each system are related to the actual \acp{isb}; the position cross-authentication, that may be used for a timing service, where the \ac{pvt} is considered to be authentic if it leads (sufficiently close) to the expected position. 
For both scenarios, we show that the attacker leads the victim to a chosen target \ac{pvt} solution, i.e., to a target 3D position or clock correction. We show that the success of the attack depends on both the number of open signals used by the receiver and the attacker conditions. Still, under realistic conditions, the attack is shown to be successful.

We considered both noiseless and noisy scenarios: in the latter, the receiver obtains a noisier replica of the signal tampered with by the receiver. Numerical results confirm the effectiveness of the attack in both the considered scenarios, showing that such a class of security checks is not able to protect the receiver effectively.

In several practical scenarios, less than 4 authenticated signals may be received so that the receiver cannot compute a \ac{pvt} based on the authenticated signals solely.
Our work shows that it is not possible to extend the authentication to open signals by using \ac{pvt}-based consistency checks, so the receiver can either choose to compute a non-secure \ac{pvt} or to not perform navigation at all. 

\bibliography{main.bib}{}
\bibliographystyle{IEEEtran}





\appendices
\section{PVT Computation with Separate Time References}
\label{app:pvt1}
We describe the \ac{pvt} linearization procedure and the derivation of each term needed to update the \ac{pvt} during the iterative approach, to compute the terms of \ref{eqn:pHr}.

From the received signal, we derive the pseudorange vector ${\bm{r}}$. For ease of reading, we consider only signals from two systems: hence the \ac{pvt} will contain entries $t_\mathrm{1}$ and $t_\mathrm{2}$, the clock biases associated with each of the time references, e.g., to the \ac{gpst} and \ac{gst}. In particular, the receiver collects $N_1$ signals from the first constellation and $N_2 = N -N_1$ from the latter. 

In general, each pseudorange is modeled by 
\begin{equation}\label{eq:rangeModel}
        r_j = \rho_j + c t_j -c t_\mathrm{1}\mathbbm{1}_\mathrm{1}(j)  - c t_\mathrm{2}\mathbbm{1}_\mathrm{2}(j) + D_{\mathrm{atm},j}+ \eta\,,
\end{equation}
where $\rho_j $ is the geometric range, i.e., the Euclidean distance between the (estimated) receiver position and the $j$th satellite position, with coordinates $(x_j,y_j,z_j)$, defined as 
\begin{equation}
    \rho_j \triangleq \sqrt{(x_j- x)^2 + (y_j- y)^2 + (z_j- z)^2}\,,
\end{equation}
$c$ is the speed of light, $t_j$ is the satellite clock bias, $D_{\mathrm{atm},j}$ models the atmospheric delays, and $\eta$ models the remaining errors, considered as noise.
The indicator functions $\mathbbm{1}_\mathrm{1}(j) $ and $\mathbbm{1}_\mathrm{2}(j) $ are introduced to identify the satellite constellations, i.e.,  $\mathbbm{1}_\mathrm{1}(j) = 1$ if satellite $j$ belongs to the first constellation and 0 otherwise, while $\mathbbm{1}_\mathrm{2}(j) = 1$ if satellite $j$ belongs to the second and is 0 otherwise.
Notice that (approximations of) the satellites' positions, $x_j$, $y_j$, and $z_j$, $t_j$, and $D_{\mathrm{atm},j}$, are retrieved directly from the navigation messages. 


Next, we introduce the terms 
\begin{equation}
    a_{x,j} \triangleq  \frac{x_j-\hat x}{\hat \rho_j}, \hspace{0.4cm}
    a_{y,j} \triangleq  \frac{y_j-\hat y}{\hat \rho_j}, \hspace{0.4cm}
    a_{z,j} \triangleq  \frac{z_j-\hat z}{\hat \rho_j}\,, 
\end{equation}
which represent the components of the unit vector pointing from the (approximate) receiver to the $j$th satellite position.
This allows us to introduce the geometry matrix 
\begin{equation}
    \bm{G} \triangleq
    \begin{bmatrix}
    a_{x,1} & a_{y,1} & a_{z,1} & 1 & 0\\
    \vdots & \vdots & \vdots & \vdots & \vdots\\
    a_{x,N_1} & a_{y,N_1} & a_{z,N_1} & 1 & 0\\
    a_{x,N_1+1} & a_{y,N_1+1} & a_{z,N_1+1} & 0 & 1\\
    \vdots & \vdots & \vdots & \vdots & \vdots\\
    a_{x,N} & a_{y,N} & a_{z,N} & 0 & 1
    \end{bmatrix}\;.
\end{equation}

Following the linearization procedure reported, for instance, in~\cite[Ch. 2]{kaplan}, and recalling \eqref{eq:rangeResiduals}, we obtain
\begin{multline}
        \Delta r_j =a_{x,j}\Delta x + a_{y,j}\Delta y + a_{z,j}\Delta z
        - \, c\,t_\mathrm{1}\mathbbm{1}_\mathrm{1}(j) - c\,t_\mathrm{2}\mathbbm{1}_\mathrm{2}(j)\,,
\end{multline}
which can be equivalently written in matrix form as \eqref{eqn:rGp}. This relates the range residuals to the \ac{pvt} update vector $\Delta\bm{p} = [\Delta x,\Delta y,\Delta z,\Delta t_\mathrm{1}, \Delta t_\mathrm{2}]^\mathrm{T}$. 

Notice that, the residuals as ordered naturally as  
\begin{align}
    \Delta\bm{r} &\triangleq 
    \begin{bmatrix}
    \Delta r_1\\
    \vdots\\
    \Delta r_{N_1}\\
    \Delta r_{N_1+1}\\
    \vdots\\
    \Delta r_{N}
    \end{bmatrix}
     =
    \begin{bmatrix}
    \Delta\bm{r}_1\\
    \Delta\bm{r}_2
    \end{bmatrix}\;.
\end{align}
Lastly, given $\bm{G}$, we can compute its Moore-Penrose inverse $\bm{H}$ and solve~\eqref{eqn:pHr}.

More details about the \ac{pvt} derivation and other aspects related to it, such as the impact of $\bm{H}$ on the solution accuracy, can be found in~\cite[Ch. 2,7]{kaplan} and \cite[Ch. 21]{teunissen2017}.

\section{PVT Computation with Unique Time Reference}\label{app:pvt2}

If the \ac{isb} between the system is known it is possible to take into account and solve the \ac{pvt} linearization considering just one of the $M$ time references. 
Without loss of generality, we consider the case where measurements are collected from two constellations and the  $\mathrm{ISB} = t_\mathrm{2} - t_\mathrm{1}$ is known by the receiver. This relates to the case where the signals are received from both GPS and Galileo and the GGTO is obtained from the I/NAV messages.  
Next, \eqref{eq:rangeModel} becomes 
\begin{equation}
        r_j = \rho_j + c t_j -c t_\mathrm{1}  - c \, \mathrm{ISB} \, \mathbbm{1}_\mathrm{2}(j) + D_{\mathrm{atm},j}+ \eta\,,
\end{equation}
Next, we can consider this as if the satellites were synchronized to the same time reference, thus matrix $\bm{G}$ becomes
\begin{equation}
    \bm{G} \triangleq
    \begin{bmatrix}
    a_{x,1} & a_{y,1} & a_{z,1} & 1 \\
    \vdots & \vdots & \vdots & \vdots \\
    a_{x,N_1} & a_{y,N_1} & a_{z,N_1} & 1 \\
    a_{x,N_1+1} & a_{y,N_1+1} & a_{z,N_1+1} & 1 \\
    \vdots & \vdots & \vdots & \vdots \\
    a_{x,N} & a_{y,N} & a_{z,N} &  1
    \end{bmatrix}\;.
\end{equation}

The rest of the procedure is analogous to the one reported in Appendix~\ref{app:pvt1}.
\end{document}